\documentclass[12p]{elsarticle}

\usepackage{hyperref}
\usepackage{booktabs}
\usepackage{nicefrac}
\usepackage[dvipsnames]{xcolor}
\usepackage{amsmath}
\usepackage{amsfonts}
\usepackage{graphicx}
\usepackage{tabularx}
\usepackage[export]{adjustbox}
\usepackage{bm}
\usepackage{bbm}
\usepackage{amsmath,amsthm,amssymb}    
\usepackage{ae}

\newcommand{\note}[1]{\textcolor{black}{#1}}
\newcommand{\noteb}[1]{\textcolor{black}{#1}}
\newcommand{\context}[1]{\textcolor{black}{#1}}
\newcommand{\challenge}[1]{\textcolor{black}{#1}}
\newcommand{\hypothesis}[1]{\textcolor{black}{#1}}
\newcommand{\proposal}[1]{\textcolor{black}{#1}}

\newcommand{\descript}[1]{\textcolor{black}{#1}}

\usepackage[english]{babel}
\usepackage{ae}
\usepackage{array}
\newtheorem{theorem}{Theorem}
\newtheorem{corollary}[theorem]{Corollary}
\newtheorem{proposition}[theorem]{Proposition}
\newtheorem{lemma}[theorem]{Lemma}
\newtheorem{definition}[theorem]{Definition}
\newtheorem{example}{Example}[section]

\newdefinition{remark}[theorem]{Remark}

\newdefinition{problem}[theorem]{Problem}
\newcommand{\li}{L([0,1])}
\def\ui{{[0,1]} }
\newcommand{\1}{\mathbbm 1}

\def\su{{\mathbf{ Su}} }
\newcommand{\sm}{\su_{m}}

\title{A Generalization of the Sugeno integral to aggregate Interval-valued data: an application to Brain Computer Interface and Social Network Analysis}

\author[pamplona]{J. Fumanal-Idocin \corref{cor1}}
\ead{javier.fumanal@unavarra.es}

\author[slovakia]{Z. Tak\'{a}\v{c}}
\ead{zdenko.takac@stuba.sk}

\author[slovakia]{\v{L}. Horansk\'a}
\ead{lubomira.horanska@stuba.sk}

\author[brasil2]{T.~da Cruz Asmus}
\ead{tiago.dacruz@unavarra.es}

\author[brasil]{G. Dimuro}
\ead{gracaliz.pereira@unavarra.es}

\author[tecnalia]{C. Vidaurre}
\ead{carmen.vidaurre@tecnalia.com}

\author[pamplona]{J. Fernandez}
\ead{fcojavier.fernandez@unavarra.es}

\author[pamplona]{H.Bustince \corref{cor1}}
\ead{bustince@unavarra.es}

\cortext[cor1]{Corresponding author}

\address[pamplona]{Public University of Navarra and Institute of Smart Cities, Campus Arrosadia s/n, 31006 Pamplona, Spain}
\address[slovakia]{Institute of Information Engineering, Automation and Mathematics, Faculty of Chemical and Food Technology, Slovak University of Technology in Bratislava, Radlinského 9, Bratislava, 812 37, Slovak Republic}
\address[brasil]{Centro de Ciencias Computacionais, Universidade Federal do Rio Grande, Av. It\'alia km 08, Campus Carreiros, Rio Grande, Brazil}
\address[brasil2]{Instituto de Matemática, Estatística e Física, Universidade Federal do Rio Grande, Av Italia km 08, Campus Carreiros, Rio Grande, Brazil}
\address[tecnalia]{Parque Científico y Tecnológico de Guipúzcoa, Paseo Mikeletegi 2, 20009 San Sebastian, Spain}

\begin{document}

\begin{abstract}
\context{Intervals are a popular way to represent the uncertainty related to data, in which we express the vagueness of each observation as the width of the interval. However, when using intervals for this purpose, we need to use the appropriate set of mathematical tools to work with.} \challenge{This can be problematic due to the scarcity and complexity of interval-valued functions in comparison with the numerical ones.}
\proposal{In this work, we propose to extend a generalization of the Sugeno integral to work with interval-valued data. Then, we use this integral to aggregate interval-valued data in two different settings: first, we study the use of intervals in a brain-computer interface; secondly, we study how to construct interval-valued relationships in a social network, and how to aggregate their information.} \note{Our results show that interval-valued data can effectively model some of the uncertainty and coalitions of the data in both cases. For the case of brain-computer interface, we found that our results surpassed the results of other interval-valued functions.}
\end{abstract}


\maketitle
\section{Introduction}
\context{
	Aggregation functions are used to fuse several data into one single output value, representative of the original set of inputs \cite{beliakov2016practical, grabisch2009aggregation}. As this fusion process is critical in many multi-component systems, aggregation functions have been studied in a wide range of different settings, like fuzzy rule-based classification systems \cite{LUCCA201894, dimuro2020state, bustince2015historical}, image processing \cite{Rudas2013, Beliakov2012}, unsupervised learning \cite{AGOP19}, decision making \cite{yager2004generalized, grabisch1996application, Rudas2013} or brain computer-interfaces \cite{fumanal2021interval, fumanal2021motor}.}

\context{Two of the most popular aggregation functions are the Choquet and  Sugeno integrals \cite{grabisch2000application, sugeno1974theory}. These integrals fuse the data using a fuzzy measure, which models the different relations among the inputs \cite{grabisch2010fuzzy, lucca2019analyzing}. By using these functions we are able to model the coalitions between the different features to fuse in the aggregation process. The famous Ordered Weighted Aggregation operators (OWAs) \cite{yager2012ordered} are a particular case of the Choquet integral using a symmetric fuzzy measure  \cite{llamazares2015constructing}.}

\context{Both the Choquet and Sugeno integrals have been profoundly studied, and many works have been published regarding their characterization,  properties, and relation with other aggregation functions \cite{murofushi1991fuzzy, dimuro2020generalized, auephanwiriyakul2002generalized, dimuro2020the}. More recently, a series of generalizations of the Choquet integral have been proposed and applied to different fuzzy-rule based classification systems \cite{dimuro2020generalized, LUCCA201894, lucca2019analyzing}. The Sugeno integral has been generalized in a similar way, applied to image processing \cite{bardozzo2021sugeno}.}

\challenge{
	 However, all of these works have been proposed to deal with numerical data. The use of other data representations, like intervals, have been proven popular to model the uncertainty linked to experimental observations or estimations, but it requires some additional challenges related to the lack of standard order in the framework of intervals \noteb{\cite{boczek2021interval}}. An Interval-Valued Sugeno integral (IV-Sugeno) was proposed in \cite{pu2019interval}, \noteb{and some more general expressions encompassing both Choquet and Sugeno integral were given in \cite{boczek2022interval}}. However, the different generalizations of the Sugeno integral were not studied.} \hypothesis{ 
	  We believe that these generalizations can be of interest in the interval-valued setting, as its numerical counterpart has been successfully applied to Brain Computer Interface systems \cite{fumanal2021motor} and image thresholding.}

\proposal{ 
	In this work, we present an in-depth study of the the properties of the IV-Sugeno, and its possible generalizations. Then, we present two different applications using interval-valued data:}

\begin{itemize}
	\item \proposal{We present a novel way to construct intervals in a Motor Imagery (MI) Brain Computer Interface (BCI) classification framework. We explain how the intervals are constructed from the output of different classifiers, and how we aggregate them using the IV-Sugeno integral.}
	\item \proposal{We present the interval-valued version of affinity functions in social network analysis \cite{idocin2020borgia}. These functions use intervals to measure the difference in commitment in a pairwise relationship between two people. Then, we explain how we use the IV-Sugeno to characterize each actor in the network based how assymetric are its relationships.}
\end{itemize}

\descript{ 
	 The rest of the paper goes as follows. Section \ref{sec:preliminary} recalls some notions needed to understand the rest of the paper: aggregation theory, the numerical Sugeno integral and interval-valued aggregations. In Section \ref{sec:iv_sugeno}, we introduce so-called  interval-valued Sugeno-like $FG$-functional as an interval-valued generalization of the Sugeno integral. In Section \ref{sec:iv_sugeno_properties}, the mathematical properties of the proposed functionals are studied, and in Section \ref{sec:iv_sugeno_construct}, a construction method for interval-valued Sugeno-like $FG$-functionals is proposed. Section \ref{sec:bci_signal} illustrates the proposed BCI framework and its use of interval-valued Sugeno-like $FG$-functionals. Section \ref{sec:social_network} proposes the interval-valued affinity functions, and four centrality measures for social network analysis using the newly proposed ideas.  Finally, in Section \ref{sec:conclusions}, we give our conclusions and future lines for this work.}

\section{Preliminary} \label{sec:preliminary}
In this section we recall some notions needed in our subsequent developments. We also fix the notation, mostly in accordance with \cite{grabisch2009aggregation}, wherein more details on  theory of aggregation functions can be found.

\subsection{Aggregation functions}
The process of merging an information represented by several values into a single one is formalized by so-called aggregation functions. The finite space of attributes can be represented by the set $N=\{1,\dots,n\}$,  $n\in\mathbb N$ and the inputs  by $n$-tuples of reals from the unit interval $\ui$. Let us denote vectors $(x_1,\dots,x_n)\in \ui^n$ by bold symbols ${\bf x}$.

\begin{definition}\label{agg}
Let $n\in \mathbb N$. A  function $A\colon \ui^n\to\ui$ is  an $n$-ary aggregation function if it is nondecreasing in each variable and the boundary conditions 
$A(0,\dots,0)=0$ and 
$A(1,\dots,1) = 1$ are satisfied.
\end{definition}

We list some of the  well-known aggregation functions. Let ${\bf x}\in \ui^n$.
\begin{itemize}
\item The {\it arithmetic mean} $\mathrm{AM}_{(n)}$ is defined by $\mathrm{AM}_{(n)}({\bf x})=\frac 1n\sum\limits_{i=1}^n x_i.$ 
\item The {\it weighted arithmetic mean} $\mathrm{WAM}_{\bf w}$ is defined by $\mathrm{WAM}_{\bf w}({\bf x})=\frac 1n\sum\limits_{i=1}^n w_i x_i,$
where the weight vector ${\bf w}=(w_1,\dots,w_n)\in[0,1]^n$ is such that  $\sum\limits_{i=1}^n w_i=1$. 
\item The {\it minimum}  and {\it maximum} operators defined by $\mathrm{Min}({\bf x})=\min\{x_1,\dots,x_n\}$ and $\mathrm{Max}({\bf x})=\max\{x_1,\dots,x_n\}$, respectively, can be regarded as the special cases of the so-called 
{\it order statistics} $\mathrm{OS}_k\colon \ui^n\to\ui$  defined by  $\mathrm{OS}_k({\bf x})=x_{\sigma(k)}$, where $\sigma$ is a permutation on $N$ such that $x_{\sigma(1)} \leq \cdots \leq x_{\sigma(n)}$. Clearly, $\mathrm{Min}=\mathrm{OS}_1$ and $\mathrm{Max}=\mathrm{OS}_n$.
\end{itemize}

The following properties of real functions can be desirable in different contexts of aggregation.
\begin{definition}\label{prop}
Let $n\in \mathbb N$. A  function $A\colon \ui^n\to\mathbb R$ is said to be:
\begin{itemize}
\item  {\it internal}, if for each ${\bf x}\in \ui^n$ it holds $\mathrm{Min}({\bf x})\leq A({\bf x})\leq \mathrm{Max}({\bf x})$;
\item {\it idempotent}, if $A(x,\dots,x)=x$ for each $x\in \ui$;
\item {\it comonotone maxitive}, if $A({\bf x}\vee{\bf \overline x})=A({\bf x})\vee A({\bf \overline x})$ for all comonotone vectors ${\bf x}, {\bf \overline x}\in \ui^n$ (vectors ${\bf x}, {\bf \overline x}$ are comonotone, if $(x_i-x_j)(\overline x_i-\overline x_j)\geq 0$  for all $i,j \in \{1, \dots, n\}$);
\item {\it comonotone minitive}, if $A({\bf x}\wedge{\bf \overline x})=A({\bf x})\wedge A({\bf \overline x})$ for all comonotone vectors ${\bf x}, {\bf \overline x}\in \ui^n$;
\item {\it positively homogeneous}, if $A(c{\bf x})=cA({\bf x})$ for each ${\bf x}\in \ui^n$ and $c>0$ such that $c{\bf x}\in \ui^n$;
\item {\it min-homogeneous}, if $A(c\wedge{\bf x})=c\wedge A({\bf x})$ for each ${\bf x}\in \ui^n$ and $c\in\ui$.
\end{itemize}
\end{definition}

\subsection{Sugeno integral}
The Sugeno integral introduced in 1974 \cite{sugeno1974theory} is widely used in many applications due to its ability to model interactions between inputs by means of the so-called fuzzy measure.
 \begin{definition}
 A  set function $m\colon 2^N\to \ui$ is  a fuzzy measure,  if $m(C)\leq m(D)$ whenever $C\subseteq D\subseteq N$ and   
$m(\emptyset)=0$,  $m(N)=1$. 
\end{definition}

 \begin{definition}
A fuzzy measure $m$ is said to be \it {symmetric}, if $m(E)=m(F)$ whenever $|E|=|F|$ for all $E,F\subseteq N$  (here $|E|$ stands for the cardinality of the set $E$).
\end{definition}
 Restricting to the finite space and vectors of reals from the unit interval, the (discrete) Sugeno integral can be defined as follows.
 \begin{definition}\label{defsug}
Let $m:\ 2^N \to [0,1]$ be a fuzzy measure, ${\bf x}=(x_1,\dots,x_n) \in [0,1]^n$.
{The Sugeno integral}  with respect to $m$ is defined by
\begin{equation}\label{form_sug}
\sm({\bf x})= \max\limits_{i\in N} \min\{x_{\sigma(i)}, m(E_{\sigma(i)})\}
\end{equation}
	 where $\sigma$ is a permutation on $N$ such that $x_{\sigma(1)} \leq \cdots \leq x_{\sigma(n)}$,   $E_{\sigma(i)}=\{\sigma(i), \ldots , \sigma(n)\}$ for $i=1,\dots,n$.
\end{definition}

Note, that the Sugeno integral is an idempotent function, which is also internal, comonotone maxitive, comonotone minitive, min-homogeneous and it extends the fuzzy measure, i.e., $\sm(\1 _A)=m(A)$,
where ${\1}_A$ is the indicator of the set $A\subseteq N$, i.e., $${\1}_A(i)=\left\{
			\begin{array}{ll}
			1 &\mbox{ if } i\in A\\[.2em]
			0 &\mbox{ otherwise }  
			\end{array}\right.
,$$ for each $i\in N$.
%
%

The following generalization of the formula (\ref{form_sug}) was introduced in \cite{bardozzo2021sugeno}, replacing the minimum and maximum operators by some more general functions.
 \begin{definition}
Let $m\colon 2^{N}\to [0,1]$ be a symmetric fuzzy measure,
		$F\colon [0,\infty[\times [0,1]\to [0,\infty[$,
		$G\colon [0,\infty[^n\to [0,\infty[$. The Sugeno-like $FG$-functional with respect to $m$ is defined by
\begin{equation}\label{bardozzo}
S_{FG}^m({\bf x})=G\left(F(x_{\sigma(1)},m(E_{\sigma(1)})),\dots,F(x_{\sigma(n)},m(E_{\sigma(n)}))\right),
\end{equation}
 with the same meaning of $\sigma$ and     $E_{\sigma(i)}$ as in Definition \ref{defsug}.
\end{definition}
Note that the symmetry of the fuzzy measure ensures the Sugeno-like $FG$-functional to be well-defined.
\subsection{Interval-Valued Aggregation functions} \label{sec:iv_agf}
When aggregating data with some uncertainty represented by intervals,  a need of interval-valued aggregation functions appears. Moreover, for using an interval-valued fuzzy integrals, there is also a need of a total order on the class of all intervals.

 Let us denote by $\li=\{X=[\underline{X},\overline{X}]\mid 0\leq \underline{X}\leq \overline{X}\leq 1\}$ the set of all closed subintervals in $[0, 1]$. We denote	$\mathbf{0}=[0,0]$, $\mathbf{1}=[1,1]$.
 The standard partial order  $\leq_{spo}$ on $L([0,1])$ is given by:
\begin{equation*}\label{usual}
{[\underline{X},\overline{X}]\leq_{spo} [\underline{Y},\overline{Y}] \ \Leftrightarrow\ \underline{X}\leq\underline{Y}\  \textrm{ and } \ \overline{X}\leq\overline{Y}}.
\end{equation*}

\begin{definition}
An order $\preceq$ on $L([0,1])$ is called admissible, if 

(i) $\preceq$ is a total order on $L([0,1])$, and

(ii) for all $X,Y\in L([0,1])$ it holds $X\preceq Y$ whenever $X\leq_{spo} Y$.
\end{definition}
The following construction method for admissible orders on $\li$ was suggested in \cite{bustince2013}.

Let	$M_1, M_2\colon [0,1]^2\to[0,1]$ be aggregation functions, such that for all $X, Y\in L([0,1])$ it holds	
	$$ M_1(\underline{X},\overline{X})=M_1(\underline{Y},\overline{Y})\wedge
	 M_2(\underline{X},\overline{X})=M_2(\underline{Y},\overline{Y})\quad\Rightarrow\quad X=Y.$$
	We define the  total order relation $\leq_{M_1,M_2}$ as follows:
      \begin{equation*}
 {X\leq_{M_1,M_2} Y  \textrm{ iff }  \begin{cases}  M_1(\underline{X},\overline{X})<M_1(\underline{Y},\overline{Y})\text{ or} \\ 
 M_1(\underline{X},\overline{X})=M_1(\underline{Y},\overline{Y}) \text{ and }  M_2(\underline{X},\overline{X})\leq M_2(\underline{Y},\overline{Y}). \end{cases}}
      \end{equation*}
Now, let   $\alpha\in[0,1]$, $X\in L([0,1])$.	Define a function $K_{\alpha}(X)=(1-\alpha)\underline{X}+\alpha \overline{X}$. 
	
	Let  $\alpha\neq\beta\in[0,1]$. The  total order relation $\leq_{M_1,M_2}=:\leq_{\alpha,\beta}$ corresponding to $M_1=K_{\alpha}$ and $M_2=K_{\beta}$ is an admissible order.

Well-known particular cases of $\leq_{\alpha,\beta}$  are the following:
\begin{itemize}
  \item  Xu and Yager order $\leq_{XY}=\leq_{0.5,1}$: 
\begin{equation*}\label{xyor}
    [\underline{X},\overline{X}]\leq_{XY} [\underline{Y},\overline{Y}] \ \Leftrightarrow \
            \begin{cases}
                    \underline{X}+\overline{X}<\underline{Y}+\overline{Y} {\text{ or }}\\
                    \underline{X}+\overline{X}=\underline{Y}+\overline{Y} \, {\text{ and }} \, \overline{X}-\underline{X}\leq \overline{Y}-\underline{Y}.
            \end{cases}
\end{equation*}
  \item lexicographical  orders $\leq_{Lex1}=\leq_{0,1}$, $\leq_{Lex2}=\leq_{1,0}$:
\begin{equation*}\label{eqlex1}
   [\underline{X},\overline{X}]\leq_{Lex_1} [\underline{Y},\overline{Y}] \ \Leftrightarrow \
            \begin{cases}
                    \underline{X}<\underline{Y} {\text{ or }} \\
                    \underline{X}=\underline{Y} \, {\text{ and }} \, \overline{X}\leq \overline{Y}
            \end{cases}
						\end{equation*}
and
\begin{equation*}\label{eqlex2}
 [\underline{X},\overline{X}]\leq_{Lex_2} [\underline{Y},\overline{Y}] \ \Leftrightarrow \
            \begin{cases}
                   \overline{X}<\overline{Y} {\text{ or }} \\
                   \overline{X}=\overline{Y} \, {\text{ and }} \, \underline{X}\leq \underline{Y}.
            \end{cases}
\end{equation*}
\end{itemize}
It was shown in \cite{bustince2013}, that for a given $\alpha\in [0, 1[$ all admissible orders $\leq_{\alpha,\beta}$ with $\beta > \alpha$ coincide. This admissible order will
be denoted by $\leq_{\alpha+}$. Similarly, for a given $\alpha\in ]0, 1]$ all admissible orders $\leq_{\alpha,\beta}$ with $\beta < \alpha$ coincide. This admissible order will
be denoted by $\leq_{\alpha-}$.

	\begin{definition}\label{def:ivfm}
Let $\preceq$ be an admissible order on $\li$. A function $m:2^{N}\to L([0,1])$ is called an interval-valued fuzzy measure  w.r.t. $\preceq$, if $m(\emptyset)=\mathbf{0}$, $m(N)=\mathbf{1}$ and $m(A)\preceq m(B)$ for all $A\subseteq B\subseteq N$.
\end{definition}
	\begin{definition}
Let $n\geq 2$. 	Let $\preceq$ be an admissible order on $\li$.
An $n$-dimensional interval-valued (IV) aggregation function w.r.t. $\preceq$ is a mapping $M:\,\li^n\to\li$ satisfying the following properties:
\begin{itemize}
  \item $M(\mathbf{0},\ldots,\mathbf{0})=\mathbf{0}$;
  \item $M(\mathbf{1},\ldots,\mathbf{1})=\mathbf{1}$;
  \item $M$ is an increasing function in each component w.r.t. $\preceq$.
\end{itemize}
\end{definition}

\section{An Interval-Valued generalization of the Sugeno Integral} \label{sec:iv_sugeno}

In this Section we give our definition of the generalized interval-valued Sugeno integral, which is an interval-valued counterpart to (\ref{bardozzo}).
To simplify our notation we will use the lattice symbols  $\mathrm{Min}=\wedge$ and $\mathrm{Max}=\vee$.

\begin{definition}\label{def:conwds}
Let $\preceq$ be an admissible order on $\li$,	$m:2^{N}\to L([0,1])$ be an IV fuzzy measure   w.r.t. $\preceq$ and $F:L([0,1]) \times L([0,1]) \to L([0,1])$, $G:\big(L([0,1])\big)^{n} \to L([0,1])$ be functions. We say that a triplet $(m,F,G)$ satisfies Condition (WDS) if for all $X_1,\ldots,X_n\in L([0,1])$ and all possible permutations $\sigma_1,\sigma_2$ on $N$ such that $X_{\sigma_1(1)}\preceq \ldots \preceq X_{\sigma_1(n)}$ and $X_{\sigma_2(1)}\preceq \ldots \preceq X_{\sigma_2(n)}$ it holds:
	\begin{multline}
	G\Big(F\big(X_{\sigma_1(1)},m(E_{\sigma_1(1)})\big), \ldots, F\big(X_{\sigma_1(n)},m(E_{\sigma_1(n)})\big)\Big) =\\ G\Big(F\big(X_{\sigma_2(1)},m(E_{\sigma_2(1)})\big), \ldots, F\big(X_{\sigma_2(n)},m(E_{\sigma_2(n)})\big)\Big),
	\end{multline}
	where $E_{\sigma_j(i)}=\{\sigma_j(i),\ldots,\sigma_j(n)\}$ for $j\in\{1,2\}$.
\end{definition}

\begin{definition}\label{def:ivsug}
	Let $n$ be a positive integer, $\preceq$ be an admissible order on $L([0,1])$ and let a triplet $(m,F,G)$ satisfies Condition (WDS). An interval-valued Sugeno-like $FG$-functional with respect to $m$ is a function $\mathbf{S}_{m}^{F,G}\colon \big(L([0,1])\big)^{n} \to L([0,1])$ given by
	\begin{equation}\label{eq:defivsug}
	\mathbf{S}_{m}^{F,G}(X_1,\ldots,X_n) = G\Big(F\big(X_{\sigma(1)},m(E_{\sigma(1)})\big), \ldots, F\big(X_{\sigma(n)},m(E_{\sigma(n)})\big)\Big)
	\end{equation}
	for all $X_1,\ldots,X_n\in L([0,1])$, where $\sigma$ is a permutation on $N$ such that $X_{\sigma(1)}\preceq \ldots \preceq X_{\sigma(n)}$ and $E_{\sigma(i)}=\{\sigma(i),\ldots,\sigma(n)\}$.
\end{definition}

The following result is immediate.

\begin{lemma}\label{lem:welldm}
	Let $g:L([0,1])\to L([0,1])$ be a function. If a triplet $(m,F,G)$ satisfies Condition (WDS), then the triplet $(m,F,g\circ G)$ satisfies Condition (WDS).
\end{lemma}

Conditions under which a triplet $(m,F,G)$ satisfies Condition (WDS), i.e. under which the function $\mathbf{S}_{m}^{F,G}$ is well defined, are given in the following proposition. 

\begin{proposition}\label{prop:welldm}
	Let $F:L([0,1]) \times L([0,1]) \to L([0,1])$, $G:\big(L([0,1])\big)^{n} \to L([0,1])$ be functions and $m:2^{N}\to L([0,1])$ be an IV fuzzy measure with respect to an admissible order $\preceq$. Then the following hold:
	\begin{enumerate} 
		\item[(i)] A triplet $(m,F,G)$ satisfies Condition (WDS) for any symmetric $m$.
		\item[(ii)] Let $G=f\circ Proj_1$ for some function $f:L([0,1])\to L([0,1])$. Then a triplet $(m,F,G)$ satisfies Condition (WDS) for any $m$.
		\item[(iii)] Let $F$ be non-decreasing in the second variable and $G=f\circ\vee$ for some function $f:L([0,1])\to L([0,1])$. Then a triplet $(m,F,G)$ satisfies Condition (WDS) for any $m$.
	\end{enumerate}
\end{proposition}
\begin{proof}
	A triplet $(m,F,G)$ satisfies Condition (WDS) if the value of the functional defined by Formula~(\ref{eq:defivsug}) is for each vector $(X_1,\ldots,X_n)\in  L([0,1])$ independent of the choice of permutation ordering this vector in accordance with given admissible order $\preceq$ on $L([0,1])$. 
	We have to discuss the case when some ties occur in vector $(X_1,\ldots,X_n)$. 
	\begin{enumerate}
		\item[(i)] Let $X_1,\ldots,X_n\in L([0,1])$ and $\sigma_1,\sigma_2$ be permutations on $N$ such that there is a tie:
		\begin{equation}
			X_{\sigma_1(k)} = \ldots = X_{\sigma_1(p)} = X_{\sigma_2(k)} = \ldots = X_{\sigma_2(p)}
		\end{equation}
		for some $1\leq k<p\leq n$, where $X_{\sigma_1(k-1)}\prec  X_{\sigma_1(k)}$ if $k>1$ and $X_{\sigma_1(p+1)}\succ  X_{\sigma_1(p)}$ if $p<n$. Then the symmetry of measure $m$ implies
		\begin{multline}
			F\Big(X_{\sigma_1(k+1)},m\big(\{\sigma_1(k+1),\ldots,\sigma_1(n)\}\big)\Big) =
			F\Big(X_{\sigma_2(k+1)},m\big(\{\sigma_2(k+1),\ldots,\sigma_2(n)\}\big)\Big)  \\
			 \vdots  \\
			F\Big(X_{\sigma_1(p)},m\big(\{\sigma_1(p),\ldots,\sigma_1(n)\}\big)\Big) = F\Big(X_{\sigma_2(p)},m\big(\{\sigma_2(p),\ldots,\sigma_2(n)\}\big)\Big)
		\end{multline}
		hence Condition (WDS) is satisfied.
		\item[(ii)] Directly follows from Equation~\eqref{eq:defivsug}.
		\item[(iii)]  Considering the same tie as in (i), by the monotonicity of $F$ and $m$, we have:
		\begin{multline}
		\vee\Bigg( F\Big(X_{\sigma_1(k)},m\big(\{\sigma_1(k),\ldots,\sigma_1(n)\}\big)\Big), \ldots, F\Big(X_{\sigma_1(p)},m\big(\{\sigma_1(p),\ldots,\sigma_1(n)\}\big)\Big)  \Bigg) = 
		\\
		F\Big(X_{\sigma_1(k)},m\big(\{\sigma_1(k),\ldots,\sigma_1(n)\}\big)\Big) = F\Big(X_{\sigma_2(k)},m\big(\{\sigma_2(k),\ldots,\sigma_2(n)\}\big)\Big) =
		\\
		\vee\Bigg( F\Big(X_{\sigma_2(k)},m\big(\{\sigma_2(k),\ldots,\sigma_2(n)\}\big)\Big), \ldots, F\Big(X_{\sigma_2(p)},m\big(\{\sigma_2(p),\ldots,\sigma_2(n)\}\big)\Big)  \Bigg).
		\end{multline}
		Hence, the triplet $(m,F,\vee)$ satisfies Condition (WDS), and consequently, by Lemma~\ref{lem:welldm}, also triplet the $(m,F,G)$ satisfies the condition.
	\end{enumerate}
\end{proof}

\begin{remark}
	Observe that the three instances of well defined triplets $(m, F, G)$ in Proposition \ref{prop:welldm} are not the only ones, but they involve all non-trivial cases. For instance, we do not mention the case of $F(X,Y)=f(X)$ for some function $f\colon L([0,1])\to L([0,1])$ which together with an arbitrary function $G\colon L([0,1])^n\to L([0,1])$ yields a well-defined function $S^{F,G}_m$ for all fuzzy measures $m$. In fact, in that case $S^{F,G}_m$ does not depend on a fuzzy measure $m$.
\end{remark}

Since in the definition of standard Sugeno integral the functions $G=\vee$ and $F=\wedge$ are used, we also consider the three special cases of IV Sugeno-like $FG$-functionals, in particular $\mathbf{S}_{m}^{\wedge,G}$, $\mathbf{S}_{m}^{F,\vee}$ and $\mathbf{S}_{m}^{\wedge,\vee}$.

\begin{corollary}\label{cor:welldm}
	Let $F\colon L([0,1]) \times L([0,1]) \to L([0,1])$, $G\colon \big(L([0,1])\big)^{n} \to L([0,1])$ be functions and $m\colon 2^{N}\to L([0,1])$ be an IV fuzzy measure with respect to an admissible order $\preceq$. Then the following hold:
	\begin{enumerate} 
		\item[(i)] Let $F$ be non-decreasing in the second variable. Then a triplet $(m,F,\vee)$ satisfies Condition (WDS) for any $m$.
		\item[(ii)] Let $G=f\circ Proj_1$ or $G=f\circ\vee$ for some function $f:L([0,1])\to L([0,1])$. Then a triplet $(m,\wedge,G)$ satisfies Condition (WDS) for any $m$. 
		\item[(iii)] A triplet $(m,\wedge,\vee)$ satisfies Condition (WDS) for any $m$.
	\end{enumerate}
\end{corollary}
\begin{proof}
	\begin{enumerate}
		\item[(i)] Directly follows from Proposition~\ref{prop:welldm} (iii).
		\item[(ii)] Directly follows from Proposition~\ref{prop:welldm} (ii)-(iii).
		\item[(iii)] Follows from item (i).
	\end{enumerate}
\end{proof}

Taking $F=\vee$, $G=\wedge$ and degenerate intervals (i.e. intervals $X=[\underline{X},\overline{X}]$ with $\underline{X}=\overline{X}$) as inputs the IV Sugeno-like $FG$-functional recovers standard Sugeno integral on $[0,1]$.

\begin{proposition}
	Let $\mu:2^{N}\to[0,1]$ be a fuzzy measure and $m:2^{N}\to L([0,1])$ be the IV fuzzy measure given by $m(A)=[\mu(A),\mu(A)]$ for all $A\subseteq N$. Let $\mathbf{S}_{m}^{\wedge,\vee}\colon \big(L([0,1])\big)^{n} \to L([0,1])$ be an IV Sugeno-like $FG$-functional with respect to $m$ for $F=\vee$ and $G=\wedge$. Then the function $f:[0,1]^{n}\to[0,1]$ given by
		\begin{equation}
	f(x_1,\ldots,x_n) = y \hspace{1.1cm} \textrm{where} \hspace{1.1cm} \mathbf{S}_{m}^{\wedge,\vee}([x_1,x_1],\ldots,[x_n,x_n]) = [y,y],
		\end{equation}
	is the Sugeno integral on $[0,1]$ with respect to $\mu$.
\end{proposition}
\begin{proof}
	The proof follows from the observation that for any admissible order $\preceq$ and all $x,y\in[0,1]$ it holds $[x,x]\preceq[y,y]$ if and only if $x\leq y$.
\end{proof}

\section{Properties of the IV Sugeno-like $FG$-functional} \label{sec:iv_sugeno_properties}

In this section we study the mathematical properties of the IV Sugeno-like $FG$-functional: idempotency, internality, positive and min-homogeneity, comonotone maxitivity and comonotone minitivity, boundary conditions, monotonicity and property of giving back the fuzzy measure. 
\subsection{Idempotency}

\begin{proposition}\label{prop:idempm}
	An interval-valued Sugeno-like $FG$-functional $\mathbf{S}_{m}^{F,G}$ is idempotent for any $m$ whenever:
	\begin{enumerate}
		\item[(i)] $F(X,\mathbf{1})=X$ for all $X\in L([0,1])$ and $G=Proj_1$; or
		\item[(ii)] $F(X,\mathbf{1})=X$ for all $X\in L([0,1])$, $F$ is non-decreasing in the second variable and $G=\vee$; or
		\item[(iii)] $G$ is idempotent and $F(X,Y)=X$ for all $X,Y\in L([0,1])$.
	\end{enumerate}
\end{proposition}
\begin{proof}
	By Proposition~\ref{prop:welldm} the functional $\mathbf{S}_{m}^{F,G}\colon \big(L([0,1])\big)^{n} \to L([0,1])$ is well-defined in all cases (i)-(iii). The fact that $\mathbf{S}_{m}^{F,G}(X,\ldots,X)=X$ for all $X\in L([0,1])$ is easy to check.
\end{proof}

\begin{remark}
	It is worth to put out that relaxing the assumption:
	\begin{center}
		$F$ is non-decreasing in the second variable
	\end{center}
	in Proposition~\ref{prop:idempm} (ii) by the assumption:
	\begin{equation}
			F(X,Y)\preceq F(X,\mathbf{1}) \,\, \textrm{ for all } X,Y\in L([0,1])
	\end{equation}

	we obtain a sufficient condition under which $\mathbf{S}_{m}^{F,\vee}$ is idempotent for any symmetric $m$.
\end{remark}

\begin{corollary}
	\begin{enumerate}
		\item[(i)] An interval-valued Sugeno-like $FG$-functional $\mathbf{S}_{m}^{F,\vee}$ is idempotent for any $m$ whenever $F(X,\mathbf{1})=X$ for all $X\in L([0,1])$ and $F$ is non-decreasing in the second variable.
		\item[(ii)] An interval-valued Sugeno-like $FG$-functional $\mathbf{S}_{m}^{\wedge,G}$ is idempotent for any $m$ whenever $G=\vee$ or $G=Proj_1$.
		\item[(iii)] An interval-valued Sugeno-like $FG$-functional $\mathbf{S}_{m}^{\wedge,\vee}$ is idempotent for any $m$.
	\end{enumerate}
\end{corollary}

Note that, as it is in item (iii) of Proposition~\ref{prop:idempm}, if $F(X,Y)=X$ for all $X,Y\in L([0,1])$, then $\mathbf{S}_{m}^{F,G}=G$ so the Sugeno-like $FG$-functional does not depend on fuzzy measure $m$ and we obtain a trivial case.

\subsection{Internality}

Recall that a function is called internal if it is between minimum and maximum. We study the conditions under which an interval-valued Sugeno-like $FG$-functional is internal for any (symmetric) fuzzy measure.

\begin{proposition}\label{prop:internalm}
	An interval-valued Sugeno-like $FG$-functional $\mathbf{S}_{m}^{F,G}$ satisfies:
	\begin{enumerate}
		\item[(i)] $\mathbf{S}_{m}^{F,G}\succeq \wedge$ for any $m$ whenever
		\begin{itemize}
			\item[(a)] $F(X,\mathbf{1})\succeq X$ for all $X\in L([0,1])$ and $G=f\circ Proj_1$ for some function $f:L([0,1])\to L([0,1])$ such that $f\succeq Id$; or
			\item[(b)] $F(X,\mathbf{1})\succeq X$ for all $X\in L([0,1])$, $F$ is non-decreasing in the second variable and $G=f\circ\vee$ for some function $f:L([0,1])\to L([0,1])$ such that $f\succeq Id$.
		\end{itemize}
		\item[(ii)] $\mathbf{S}_{m}^{F,G}\preceq \vee$ for any $m$ whenever 
		\begin{itemize}
			\item[(a)] $F(X,\mathbf{1})\preceq X$ for all $X\in L([0,1])$ and $G=f\circ Proj_1$ for some function $f:L([0,1])\to L([0,1])$ such that $f\preceq Id$; or
			\item[(b)] $F(X,\mathbf{1})\preceq X$ for all $X\in L([0,1])$, $F$ is non-decreasing in the second variable and $G=f\circ\vee$ for some function $f:L([0,1])\to L([0,1])$ such that $f\preceq Id$.
		\end{itemize}
		\item[(iii)] $\mathbf{S}_{m}^{F,G}$ is internal for any $m$ whenever
		\begin{itemize}
			\item[(a)] $F(X,\mathbf{1}) = X$ for all $X\in L([0,1])$ and $G=Proj_1$; or
			\item[(b)] $F(X,\mathbf{1}) = X$ for all $X\in L([0,1])$, $F$ is non-decreasing in the second variable and $G=\vee$.
		\end{itemize}
	\end{enumerate}
\end{proposition}
\begin{proof}
	By Proposition~\ref{prop:welldm} and Corollary~\ref{cor:welldm} the functional $\mathbf{S}_{m}^{F,G}\colon \big(L([0,1])\big)^{n} \to L([0,1])$ is well-defined in all cases. The fact that $\mathbf{S}_{m}^{F,G}$ is internal is easy to check in each case separately.
\end{proof}

\begin{remark}
	It is worth to put out that relaxing the assumption:
	\begin{center}
		$F$ is non-decreasing in the second variable
	\end{center}
	in Proposition~\ref{prop:internalm} (i)(b) (or Proposition~\ref{prop:internalm} (ii)(b), or Proposition~\ref{prop:internalm} (iii)(b)) by the assumption:
	\begin{equation}
		F(X,Y)\preceq F(X,\mathbf{1}) \,\, \textrm{ for all } X,Y\in L([0,1])
	\end{equation}
	we obtain a sufficient condition under which $\mathbf{S}_{m}^{F,G}\succeq \wedge$ (or $\mathbf{S}_{m}^{F,G}\preceq \vee$, or $\mathbf{S}_{m}^{F,G}$ is internal) for any symmetric $m$.
\end{remark}

\begin{corollary}\label{cor:internalm}
	\begin{enumerate}
		\item[(i)] An interval-valued Sugeno-like $FG$-functional $\mathbf{S}_{m}^{F,\vee}$ satisfies:
		\begin{enumerate}
			\item[(a)] $\mathbf{S}_{m}^{F,\vee}\succeq \wedge$ for any $m$ whenever $F(X,\mathbf{1})\succeq X$ for all $X\in L([0,1])$ and $F$ is non-decreasing in the second variable.
			
			\item[(a$^s$)] $\mathbf{S}_{m}^{F,\vee}\succeq \wedge$ for any symmetric $m$ whenever $F(X,\mathbf{1})\succeq X$ for all $X\in L([0,1])$.
			
			\vspace{0.1cm}
			
			\item[(b)] $\mathbf{S}_{m}^{F,\vee}\preceq \vee$ for any $m$ whenever $F(X,\mathbf{1})\preceq X$ for all $X\in L([0,1])$ and $F$ is non-decreasing in the second variable.
			
			\item[(b$^s$)] $\mathbf{S}_{m}^{F,\vee}\preceq \vee$ for any symmetric $m$ whenever $F(X,Y)\preceq X$ for all $X,Y\in L([0,1])$.
			
			\vspace{0.1cm}
			
			\item[(c)] $\mathbf{S}_{m}^{F,\vee}$ is internal for any $m$ whenever $F(X,\mathbf{1})=X$ for all $X\in L([0,1])$ and $F$ is non-decreasing in the second variable.
			
			\item[(c$^s$)] $\mathbf{S}_{m}^{F,\vee}$ is internal for any symmetric $m$ whenever $F(X,Y)\preceq X$ and $F(X,\mathbf{1})=X$ for all $X\in L([0,1])$.
		\end{enumerate}
		
		\item[(ii)] An interval-valued Sugeno-like $FG$-functional $\mathbf{S}_{m}^{\wedge,G}$ satisfies:
		\begin{enumerate}
			\item[(a)] $\mathbf{S}_{m}^{\wedge,G}\succeq \wedge$ for any $m$ whenever $G=f\circ Proj_1$ or  $G=f\circ\vee$ for some function $f\colon L([0,1])\to L([0,1])$ such that $f\succeq Id$.
			
			\vspace{0.1cm}
			
			\item[(b)] $\mathbf{S}_{m}^{\wedge,G}\preceq \vee$ for any $m$ whenever $G=f\circ Proj_1$ or  $G=f\circ\vee$ for some function $f\colon L([0,1])\to L([0,1])$ such that $f\preceq Id$.
			
			\vspace{0.1cm}
			
			\item[(c)] $\mathbf{S}_{m}^{\wedge,G}$ is internal for any $m$ whenever $G=Proj_1$ or  $G=\vee$.
		\end{enumerate}
		
		\item[(iii)] An interval-valued Sugeno-like $FG$-functional $\mathbf{S}_{m}^{\wedge,\vee}$ is internal for any $m$.
	\end{enumerate}
\end{corollary}

\subsection{Positive and min-homogeneity}

Recall that a function $f\colon L([0,1])^n\to L([0,1])$ is called positively homogeneous if 
$f\left(cX_1,\dots,cX_n\right)=cf(X_1,\dots,X_n)$ for all $(X_1,\dots,X_n)\in L([0,1])^n$ and all $c\in \mathbb R^+$ such that $c(X_1,\dots,X_n)=(cX_1,\dots,cX_n)\in L([0,1])^n$  with convention $c [\underline{X},\overline{X}]=[c\underline{X},c\overline{X}]$.

Similarly, a function $f\colon L([0,1])^n\to L([0,1])$ is called min-homogeneous if 
$f(c\wedge (X_1,\dots,X_n))=c\wedge f(X_1,\dots,X_n)$ for all $(X_1,\dots,X_n)\in L([0,1])^n$ and all $c\in \mathbb R^+$ such that $c\wedge (X_1,\dots,X_n)=(c\wedge X_1,\dots, c\wedge X_n)\in L([0,1])^n$, with convention $c\wedge [\underline{X},\overline{X}]=[c\wedge\underline{X},c\wedge\overline{X}]$.
We study the conditions under which an interval-valued Sugeno-like $FG$-functional is positively (min-)homogeneous for any (symmetric) fuzzy measure.

\begin{proposition}\label{prop:homogeneity}
	An interval-valued Sugeno-like $FG$-functional $\mathbf{S}_{m}^{F,G}$ is positively (min-)homogeneous for any $m$ whenever $G$ is positively (min-)homogeneous and $F(\cdot,y)$  is positively (min-)homogeneous for every $y\in\li$.
\end{proposition}
\begin{proof}
	Easy to check.
\end{proof}

\begin{corollary}
	\begin{enumerate}
		\item[(i)] An interval-valued Sugeno-like $FG$-functional $\mathbf{S}_{m}^{F,\vee}$ is positively (min-)homogeneous for any $m$ whenever  $F(\cdot,y)$  is positively (min-)homogeneous for every $y\in\li$.
		\item[(ii)] An interval-valued Sugeno-like $FG$-functional $\mathbf{S}_{m}^{\wedge,G}$ is  min-homogeneous for any $m$ whenever $G$ is min-homogeneous.
		\item[(iii)] An interval-valued Sugeno-like $FG$-functional $\mathbf{S}_{m}^{\wedge,\vee}$ is  min-homogeneous for any $m$.
	\end{enumerate}
\end{corollary}

\subsection{Comonotone maxitivity}
Recall that a function $f\colon L([0,1])^n\to L([0,1])$ is called comonotone maxitive if 
$f\left(X\vee Y)\right)=f(X)\vee f(Y)$ for all  $X=(X_1,\dots,X_n)\in L([0,1])^n$ and $Y=(Y_1,\dots,Y_n)\in L([0,1])^n$ which are comonotone, i.e., there exists a permutation $\sigma$ on $N$ with $X_{\sigma(1)}\preceq \ldots \preceq X_{\sigma(n)}$ and $Y_{\sigma(1)}\preceq \ldots \preceq Y_{\sigma(n)}$.

\begin{proposition}\label{prop:maxitivity}
	An interval-valued Sugeno-like $FG$-functional $\mathbf{S}_{m}^{F,G}$ is comonotone maxitive for any $m$ whenever $G$ is comonotone maxitive and $F(\cdot,y)$  is comonotone maxitive for every $y\in\li$.
\end{proposition}
\begin{proof}
	Easy to check.
\end{proof}

\begin{corollary}
	\begin{enumerate}
		\item[(i)] An interval-valued Sugeno-like $FG$-functional $\mathbf{S}_{m}^{F,\vee}$ is comonotone maxitive for any $m$ whenever  $F(\cdot,y)$  is comonotone maxitive for every $y\in\li$.
		\item[(ii)] An interval-valued Sugeno-like $FG$-functional $\mathbf{S}_{m}^{\wedge,G}$ is comonotone maxitive for any $m$ whenever $G$ is comonotone maxitive.
		\item[(iii)] An interval-valued Sugeno-like $FG$-functional $\mathbf{S}_{m}^{\wedge,\vee}$ is  comonotone maxitive for any $m$.
	\end{enumerate}
\end{corollary}

\subsection{Boundary conditions and monotonicity}

\begin{proposition}\label{prop:boundary0}
	An interval-valued Sugeno-like $FG$-functional $\mathbf{S}_{m}^{F,G}\colon \big(L([0,1])\big)^{n} \to L([0,1])$ fulfills condition $\mathbf{S}_{m}^{F,G}(\mathbf{0},\dots,\mathbf{0})=\mathbf{0}$ for any $m$ whenever:
	\begin{enumerate}
		\item[(i)] $F(\mathbf{0},\mathbf{1})=\mathbf{0}$, $G=f\circ Proj_1$ and $f(\mathbf{0})=\mathbf{0}$; or
		\item[(ii)] $F(\mathbf{0},Y)=\mathbf{0}$ for all $Y\in L([0,1])$ and   $G=f\circ\vee$ and $f(\mathbf{0})=\mathbf{0}$
	\end{enumerate}
\end{proposition}
\begin{proof}
	By Proposition~\ref{prop:welldm} the functional $\mathbf{S}_{m}^{F,G}$ is well-defined in  cases (i)-(ii). The fact that $\mathbf{S}_{m}^{F,G}(\mathbf{0},\dots,\mathbf{0})=\mathbf{0}$ is easy to check.
\end{proof}
\begin{proposition}\label{prop:boundary0Sym}
	An interval-valued Sugeno-like $FG$-functional $\mathbf{S}_{m}^{F,G}\colon \big(L([0,1])\big)^{n} \to L([0,1])$ fulfills condition $\mathbf{S}_{m}^{F,G}(\mathbf{0},\dots,\mathbf{0})=\mathbf{0}$ for any symmetric $m$ whenever $G(\mathbf{0},\dots,\mathbf{0})=\mathbf{0}$  and $F(\mathbf{0},Y)=\mathbf{0}$ for all $Y\in L([0,1])$.
\end{proposition}
\begin{proposition}\label{prop:boundary1}
	An interval-valued Sugeno-like $FG$-functional $\mathbf{S}_{m}^{F,G}\colon \big(L([0,1])\big)^{n} \to L([0,1])$ fulfills condition $\mathbf{S}_{m}^{F,G}(\mathbf{1},\dots,\mathbf{1})=\mathbf{1}$ for any $m$ whenever:
	\begin{enumerate}
		\item[(i)] $F(\mathbf{1},\mathbf{1})=\mathbf{1}$, $G=f\circ Proj_1$ and $f(\mathbf{1})=\mathbf{1}$; or
		\item[(ii)] $F(\mathbf{1},\mathbf{1})=\mathbf{1}$ and $F$ is non-decreasing in the second variable,  $G=f\circ\vee$ and $f(\mathbf{1})=\mathbf{1}$
	\end{enumerate}
\end{proposition}
\begin{proposition}\label{prop:boundary1Sym}
	An interval-valued Sugeno-like $FG$-functional $\mathbf{S}_{m}^{F,G}\colon \big(L([0,1])\big)^{n} \to L([0,1])$ fulfills condition $\mathbf{S}_{m}^{F,G}(\mathbf{1},\dots,\mathbf{1})=\mathbf{1}$ for any symmetric $m$ whenever $G(\mathbf{1},\dots,\mathbf{1})=\mathbf{1}$  and $F(\mathbf{1},Y)=\mathbf{1}$ for all $Y\in L([0,1])$.
\end{proposition}

\begin{proposition}\label{prop:monotonicity}
	An interval-valued Sugeno-like $FG$-functional $\mathbf{S}_{m}^{F,G}\colon \big(L([0,1])\big)^{n} \to L([0,1])$ is non-decreasing for any $m$ whenever:
	\begin{enumerate}
		\item[(i)] $F(\cdot,\mathbf{1})$ is  non-decreasing,  $G=f\circ Proj_1$ and $f$ is non-decreasing; or
		\item[(ii)] $F$ is  non-decreasing,   $G=f\circ\vee$ and $f$ is non-decreasing.
	\end{enumerate}
\end{proposition}

\begin{proposition}\label{prop:monotonicity}
	An interval-valued Sugeno-like $FG$-functional $\mathbf{S}_{m}^{F,G}\colon \big(L([0,1])\big)^{n} \to L([0,1])$ is non-decreasing for any symmetric $m$ whenever
	$G$ is non-decreasing and $F$ is non-decreasing in the first variable.
	
\end{proposition}
\begin{proof}
	By Proposition~\ref{prop:welldm} the functional $\mathbf{S}_{m}^{F,G}$ is well-defined. The claim is easy to check.
\end{proof}

\begin{corollary}\label{prop:agg}
	An interval-valued Sugeno-like $FG$-functional $\mathbf{S}_{m}^{F,G}$ is an aggregation function for any $m$ whenever
	\begin{enumerate}
		\item[(i)]   $F(\cdot,\mathbf{1})$ is  non-decreasing, $F(\mathbf{0},\mathbf{1})=\mathbf{0}$,  $F(\mathbf{1},\mathbf{1})=\mathbf{1}$, $G=f\circ Proj_1$ and $f$ is non-decreasing with $f(\mathbf{0})=\mathbf{0}$, $f(\mathbf{1})=\mathbf{1}$; or
		\item[(ii)] $F(\mathbf{0},Y)=\mathbf{0}$ for all $Y\in L([0,1])$, $F(\mathbf{1},\mathbf{1})=\mathbf{1}$ and $F$ is non-decreasing,  and   $G=f\circ\vee$ and $f$ is non-decreasing with $f(\mathbf{0})=\mathbf{0}$,  $f(\mathbf{1})=\mathbf{1}$.
	\end{enumerate}
\end{corollary}

\begin{corollary}\label{prop:aggsym}
	An interval-valued Sugeno-like $FG$-functional $\mathbf{S}_{m}^{F,G}$ is an aggregation function for any symmetric $m$ whenever
	$G$ is  non-decreasing with $G(\mathbf{0},\dots,\mathbf{0})=\mathbf{0}$, $G(\mathbf{1},\dots,\mathbf{1})=\mathbf{1}$  and $F$ is non-decreasing in the first variable with
	$F(\mathbf{0},Y)=\mathbf{0}$ 
	and $F(\mathbf{1},Y)=\mathbf{1}$ for all $Y\in L([0,1])$.
\end{corollary}

\begin{example}
	Let $G(X_1,\dots,X_n)=\frac 1n\sum_{i=1}^n X_i$; $F(X,Y)=X^2\cdot Y+X\cdot (1-Y)$ and $m(A)=\left(\frac{|A|}{n}\right)^2$ for $A\subseteq N$. Then 
	$S^{F,G}_m(X_1,\dots,X_n)=\frac 1n\sum_{i=1}^n X_i+\frac 1n\sum_{i=1}^n \left(\frac{n-i}{n}\right)^2\left(X_i^2-X_i\right)$ which is   an aggregation function according to Corollary \ref{prop:aggsym}.
\end{example}

\subsection{Giving back the fuzzy measure}
Recall that one of the important properties of the standard Sugeno integral is that it gives back the fuzzy measure, i.e. $\mathbf{Su}_m(\1_E)=m(E)$ for all $E\subseteq N$, where $\1_E$ is the indicator of the set $E$.  A modification of this property for an interval-valued Sugeno-like $FG$-functional $\mathbf{S}_{m}^{F,G}$ can be formulated as follows: 

$\mathbf{S}_{m}^{F,G}(\1_E)=F\left(\mathbf{1},m(E)\right)$ for all $E\subseteq N$, where $\1_E\in L([0,1])^n$ is the interval-valued indicator of the set $E$ defined by $\1_E(i)=\mathbf{1}$   if  $i\in E$ and $\1_E(i)=\mathbf{0}$ otherwise.

Obviously, this property cannot be satisfied  in the case of $G=f\circ Proj_1$ nor in the case when we consider only symmetric measures and general $G$'s. 

\begin{proposition}\label{prop:givingcap}
	An interval-valued Sugeno-like $FG$-functional $\mathbf{S}_{m}^{F,\vee}$ gives back the fuzzy measure for any $m$ whenever 
	$F(\mathbf{0},\mathbf{1})=\mathbf{0}$ and $F$ is non-decreasing in the second variable.
\end{proposition}
\begin{proof}
	Let $E$ be subset of $N$ with $\mathrm{card}(E)=k$. Then $(\1_E)_{\sigma(i)}=\mathbf{0}$ for $i=1,\dots,n-k$,
	$(\1_E)_{\sigma(i)}=\mathbf{1}$ for $i=n-k+1,\dots,n$ and $E_{\sigma(n-k+1)}=E$, where $\sigma$ is a permutation ordering the vector $\1_E$ non-decreasingly. Taking into account non-decreasingness of $F$ in the second variable, we get
	\begin{multline}
		\mathbf{S}_{m}^{F,\vee}(\1_E)=
		\bigvee\limits_{i=1}^{n-k} F(\mathbf{0}, m(E_{\sigma(i)}))\vee\bigvee\limits_{i=n-k+1}^n F(\mathbf{1}, m(E_{\sigma(i)}))=
		F(\mathbf{0},\mathbf{1})\vee F(\mathbf{1}, m(E)),
	\end{multline}
	and the claim follows.
\end{proof}

\begin{table}
\scalebox{0.7}{
   \hskip-0.7cm 
	\setlength\extrarowheight{8.5pt}
    \begin{tabular}{c||c|c |c}
        Property & $G=f\circ \vee$, $F \nearrow$ in sec.var.  & $G=f\circ Proj_1$ & $m$ symmetric \\
        \hline
        \hline
       
        $\mathbf{S}_{m}^{F,G}(\mathbf{0},\ldots,\mathbf{0})=\mathbf{0}$ & $F(\mathbf{0},\cdot)=\mathbf{0}$, $f(\mathbf{0})=\mathbf{0}$  & $F(\mathbf{0},\mathbf{1})=\mathbf{0}$,  $f(\mathbf{0})=\mathbf{0}$ & $F(\mathbf{0},\cdot)=\mathbf{0}$, $G(\mathbf{0},\dots,\mathbf{0})=\mathbf{0}$\\
        \hline
        $\mathbf{S}_{m}^{F,G}(\mathbf{1},\ldots,\mathbf{1})=\mathbf{1}$  & $F(\mathbf{1},\mathbf{1})=\mathbf{1}$, $f(\mathbf{1})=\mathbf{1}$  &$F(\mathbf{1},\mathbf{1})=\mathbf{1}$,  $f(\mathbf{1})=\mathbf{1}$ & $F(\mathbf{1},\cdot)=\mathbf{1}$, $G(\mathbf{1},\dots,\mathbf{1})=\mathbf{1}$\\
        \hline
        Non-decreasingness & $f\nearrow$   &$f\nearrow$   & $G\nearrow$,  $F\nearrow$ in first var. \\
        \hline
        Aggregation function &the above lines in this table   & the above lines in this table  &the above lines in this table \\
        \hline
        Idempotency & $f=\mathrm{id}$  & $f=\mathrm{id}$ & $G$ idempotent, $F(X,Y)=X$\\
        \hline
        Internality & $f=\mathrm{id}$, $F(X,{\bf 1})=X$   & $f=\mathrm{id}$, $F(X,{\bf 1})=X$  &--\\
        \hline
        Pos. homogeneity & $f$ pos. hom., $F$ pos. hom.  &$f$ pos. hom., $F$ pos. hom.   &$G$ pos. hom., $F$ pos. hom.  \\
        \hline
        Min-homogeneity &  $f$ min-hom., $F$ min-hom.    &  $G$ min-hom., $F$ min-hom.  & $G$ min-hom., $F$ min-hom.  \\
        \hline
        Comonotone maxitivity & $F$ com. maxitive    & $F$ com. maxitive    & $G$ com.max., $F$ com. max. \\
        \hline
        Giving back the fuzzy measure & $f=\mathrm{id}$,  $F(\mathbf{0},\mathbf{1})=\mathbf{0}$  &  --&--\\
        \hline
    \end{tabular}
    }
	    	\caption{Summary of sufficient conditions for particular properties of the Sugeno-like $FG$-functionals in three cases of Prop.\ref{prop:welldm}. Symbol $\nearrow$ denotes non-decreasingness.}
	\label{tab:propsug}
\end{table}
\section{Construction of IV Sugeno-like FG-functional} \label{sec:iv_sugeno_construct}
 In this section we construct IV Sugeno-like $FG$-functionals with respect to the orders $\leq_{\alpha+}$ and $\leq_{\alpha-}$. Since the width of the input intervals can be regarded as the measure of data uncertainty represented by them, it is desirable to take it into account in our construction. We use the similar approach as proposed in \cite{asiain2017negations}.

\begin{definition}[\cite{asiain2017negations}]
	Let $c\in [0,1]$, $\alpha\in [0,1]$ and $Z=[x,y]\in L([0,1])$. We denote  by $w(Z)=y-x$ and by $d_{\alpha}(c)$ the maximal possible width of an interval $X\in  L([0,1])$ such that $K_{\alpha}(X)=c$. Moreover, we define
	\begin{equation}
		\lambda_{\alpha}(Z)=\frac{w(Z)}{d_{\alpha}(K_{\alpha}(Z))}
	\end{equation}
	where we set $\frac{0}{0}=1$.
\end{definition}

\begin{proposition}[\cite{asiain2017negations}]
	For all $\alpha\in [0,1]$ and $X\in L([0,1])$ it holds that
	\begin{equation}
	d_{\alpha}(K_{\alpha}(X))=\wedge\left(\frac{K_{\alpha}(X)}{\alpha}, \frac{1- K_{\alpha}(X)}{1-\alpha}\right).
		\end{equation}
	where we set $\frac{r}{0}=1$ for all $r\in[0,1]$.
\end{proposition}

Let $\alpha\in[0,1]$ and let $M_1, M_2:[0,1]^n\to[0,1]$ be $n$-ary functions. We define an interval function $M_{IV}:(L([0,1]))^n  \rightarrow L([0,1])$ as follows:
\begin{equation}\label{eq:conmiv}
M_{IV}(X_1,\ldots,X_n)=Y, \quad\textrm{ where } \qquad
\begin{cases}
K_{\alpha}(Y)=M_1\left(K_{\alpha}(X_1),\ldots,K_{\alpha}(X_n)\right),
\\[0.1cm]
\lambda_{\alpha}(Y)=M_2\left(\lambda_{\alpha}(X_1),\ldots,\lambda_{\alpha}(X_n)\right),
\end{cases}
\end{equation}
for all $X_1,\ldots,X_n\in L([0,1])$. It was proved in \cite{asiain2017negations} that $M_{IV}$ is an IV aggregation function with respect to $\leq_{\alpha+}$ and $\leq_{\alpha-}$ whenever $M_1, M_2$ be aggregation functions where $M_1$ is strictly increasing. Now we show that the obtained $M_{IV}$ preserves some properties of the functions $M_1$ and $M_2$, in particular the properties of $t$-norms.

\begin{proposition}\label{prop:propmiv}
	Let $\alpha\in[0,1]$, let $M_1, M_2:[0,1]^n\to[0,1]$ be functions and let $M_{IV}:(L([0,1]))^n  \rightarrow L([0,1])$ be an interval function given by Equation~\eqref{eq:conmiv}. Then the following hold:
	\begin{enumerate}
		\item[(i)] $M_{IV}(\mathbf{0},\ldots,\mathbf{0})=\mathbf{0}$ whenever $M_1(0,\ldots,0)=0$.
		\item[(ii)] $M_{IV}(\mathbf{1},\ldots,\mathbf{1})=\mathbf{1}$ whenever $M_1(1,\ldots,1)=1$.
		\item[(iii)] For $n=2$ $M_{IV}(\mathbf{0},Y)=\mathbf{0}$ for all $Y\in L([0,1])$ whenever $M_1(0,y)=0$ for all $y\in [0,1]$.
		\item[(iv)] For $n=2$ $M_{IV}(\mathbf{0},Y)=Y$ for all $Y\in L([0,1])$ whenever $M_1(0,y)=y$ and $M_2(1,y)=y$ for all $y\in [0,1]$.
		\item[(v)] For $n=2$ $M_{IV}(\mathbf{1},Y)=Y$ for all $Y\in L([0,1])$ whenever $M_1(1,y)=y$ for all $y\in [0,1]$.
		\item[(vi)] $M_{IV}$ is symmetric whenever $M_1$ and $M_2$ are symmetric.
		\item[(vii)] $M_{IV}$ is associative whenever $M_1$ and $M_2$ are associative.
		\item[(viii)] $M_{IV}$ is non-decreasing whenever $M_1,M_2$ are non-decreasing and $M_1$ is strictly increasing on
			\begin{equation}
		\{(x_1,\ldots,x_n)\in[0,1]^n\,|\,M_1(x_1,\ldots,x_n)\notin\{0,1\}\}.
			\end{equation}
		\item[(ix)] Let $k\in N$, then $M_{IV}$ is non-decreasing in the $k$-th variable whenever $M_1,M_2$ are non-decreasing in the $k$-th variable and $M_1$ is strictly increasing in the $k$-th variable on
			\begin{equation}
		\{(x_1,\ldots,x_n)\in[0,1]^n\,|\,M_1(x_1,\ldots,x_n)\notin\{0,1\}\}.
			\end{equation}
	\end{enumerate}
\end{proposition}
\begin{proof}
	The items (i), (ii) and (viii) follows from \cite[Theorem 3.16]{bustince2013}. With respect to (iv) observe:
		\begin{equation}
	K_{\alpha}\big(M_{IV}(\mathbf{0},Y)\big) = M_1\big(K_{\alpha}(\mathbf{0}),K_{\alpha}(Y)\big) = M_1\big(0,K_{\alpha}(Y)\big) = K_{\alpha}(Y)
		\end{equation}
	and
		\begin{equation}
	\lambda_{\alpha}\big(M_{IV}(\mathbf{0},Y)\big) = M_2\big(\lambda_{\alpha}(\mathbf{0}),\lambda_{\alpha}(Y)\big) = M_1\big(1,\lambda_{\alpha}(Y)\big) = \lambda_{\alpha}(Y).
		\end{equation}
	The proofs of (iii) and (v) are similar to that of (iv) observing that if $K_{\alpha}(Y)=0$ (or $K_{\alpha}(Y)=1$), then $Y=\mathbf{0}$ (or $Y=\mathbf{1}$) regardless the value of $\lambda_{\alpha}(Y)$. The items (vi) and (ix) are straightforward. Finally, with respect to (vii), let $U=M_{IV}\big(M_{IV}(X,Y),Z\big)$ and $V=M_{IV}\big(X,M_{IV}(Y,Z)\big)$. Then, by the associativity of $M_1$, we have
	\begin{multline}
			K_{\alpha}(U) = M_1\Big(M_1\big(K_{\alpha}(X),K_{\alpha}(Y)\big),K_{\alpha}(Z)\Big) =\\ M_1\Big(K_{\alpha}(X),M_1\big(K_{\alpha}(Y),K_{\alpha}(Z)\big)\Big) = K_{\alpha}(V)
	\end{multline}
	and similarly, by the associativity of $M_2$, also $\lambda_{\alpha}(U)=\lambda_{\alpha}(V)$.
\end{proof}

\begin{remark}
	It is worth to put out that, taking $M_1=\wedge$ (i.e., minimum in $[0,1]$), the induced $M_{IV}$ is not minimum in $L([0,1])$ for any $M_2$. In particular, it is easy to check that since minimum in $[0,1]$ is not strictly increasing, the monotonicity of $M_{IV}$ is violated. Hence, $\mathbf{S}_{m}^{\wedge,\vee}\neq \mathbf{S}_{m}^{M_{IV},\vee}$ if $M_{IV}$ is induced by $M_1=\wedge$.
\end{remark}

\begin{corollary}
	Under the assumptions of Proposition~\ref{prop:propmiv}, if $n=2$ then the following hold:
	\begin{enumerate}
		\item[(i)] $M_{IV}$ is an IV $t$-norm whenever $M_1$ is a strictly increasing $t$-norm and $M_2$ is a commutative, associative and non-decreasing function.
		\item[(ii)] $M_{IV}$ is an IV $t$-conorm whenever $M_1$ is a strictly increasing $t$-conorm and $M_2$ is $t$-norm.
	\end{enumerate}
\end{corollary}

\begin{example}\label{ex:conivsug}
	We give examples of IV Sugeno-like $FG$-functionals $\mathbf{S}_{m}^{F,G}$ obtained by the construction based on Proposition~\ref{prop:welldm}, Equation~\eqref{eq:conmiv} and Proposition~\ref{prop:propmiv}. The measure $m$ need not be symmetric. The functional $\mathbf{S}_{m}^{F,G}$ is well-defined if we take  $G=f\circ\vee$ where, for instance, $f(X)=X$ or $f(X)=X^2=\left[\underline{X}^2,\overline{X}^2\right]$ or $f(X)=\sqrt{X}=\left[\sqrt{\underline{X}},\sqrt{\overline{X}}\right]$; and $F$ induced by $M_1,M_2\colon[0,1]^2\to[0,1]$ where, for instance:
	\begin{enumerate}
		\item[(i)] $M_1(x,y)=xy$ and $M_2=M_1$.
		\vspace{0.1cm}    
		\item[(ii)] $M_1(x,y)=xy$ and $M_2(x,y)=(1-a_2)x+a_2y$ for some $a_2\in[0,1]$.
		\vspace{0.1cm}    
		\item[(iii)] $M_1(x,y)=\begin{cases}
		0, & \textrm{if }  x=y=0\\
		\frac{xy}{x+y-xy}, & \textrm{otherwise}
		\end{cases}
		$
		\hspace{0.5cm} and $M_2=M_1$.
		\vspace{0.1cm}    
		\item[(iv)] $M_1(x,y)=\begin{cases}
		0, & \textrm{if }  x=y=0\\
		\frac{xy}{x+y-xy}, & \textrm{otherwise}
		\end{cases}
		$
		\hspace{0.5cm} and $M_2(x,y)=(1-a_2)x+a_2y$ for some $a_2\in[0,1]$.
		\vspace{0.1cm}    
		\item[(v)] $M_1(x,y)=(1-a_1)x+a_1y$ and $M_2(x,y)=(1-a_2)x+a_2y$ for some $a_1\in]0,1[$ and $a_2\in[0,1]$. 
	\end{enumerate}
\end{example}
	Note that:
	\begin{itemize}
		\item In all the relevant cases, (ii), (iv) and (v), by the parameters $a_1,a_2$ we can regulate the relative weight we put on inputs $X_{\sigma(i)}$ and fuzzy measures $m(E_{\sigma(i)})$, in particular, taking $a_1=0.5$ we put the same weight to $X_{\sigma(i)}$ as to $m(E_{\sigma(i)})$; taking $a_1=0.25$ we put the greater weight to $X_{\sigma(i)}$ as to $m(E_{\sigma(i)})$ and vice versa for $a_1=0.75$.
		\item The measure $m$ need not to be symmetric in the above examples (i)-(v). For symmetric measures, there are no restrictions for $F$ and $G$, see Proposition~\ref{prop:welldm} (i), so we can apply any functions $F:L([0,1]) \times L([0,1]) \to L([0,1])$, $G:\big(L([0,1])\big)^{n} \to L([0,1])$. 
		\item The function $M_1$ in items (iii)-(iv) is the Hamacher product so it is possible to generalize the class of examples by taking $M_1(x,y)=\frac{xy}{\gamma+(1-\gamma)(x+y-xy)}$ for any $\gamma\geq 0$. In fact, we can apply as $M_1$ any strict $t$-norm.
		\item In Table~\ref{tab:propivsug} the summary of the properties satisfied by the Sugeno-like $FG$-functionals constructed in the above examples (i)-(v) is given.
	\end{itemize}

\begin{table}
	\scalebox{0.9}{
\centering
	\setlength\extrarowheight{8.5pt}
		\begin{tabular}{c||c|c}
		Property & Sugeno-like $FG$-functionals & Justification  \\
		\hline
		\hline
		Well-defined & (i), (ii), (iii), (iv), (v)  & Proposition~\ref{prop:welldm} (iii) \\
		\hline
		$\mathbf{S}_{m}^{F,G}(\mathbf{0},\ldots,\mathbf{0})=\mathbf{0}$ & (i), (ii), (iii), (iv)  & Proposition~\ref{prop:propmiv} (iii) \\
		\hline
		$\mathbf{S}_{m}^{F,G}(\mathbf{1},\ldots,\mathbf{1})=\mathbf{1}$ & (i), (ii), (iii), (iv), (v)  & Proposition~\ref{prop:propmiv} (ii) \\
		\hline
		Non-decreasingness & (i), (ii), (iii), (iv), (v)  & Proposition~\ref{prop:propmiv} (viii) \\
		\hline
		Aggregation function & (i), (ii), (iii), (iv)  & The above lines of this table \\
		\hline
		Idempotency & (i), (iii), only for $G=\vee$  & Proposition~\ref{prop:idempm} (i) \\
		\hline
		Internality & (i), (iii), only for $G=\vee$  & Proposition~\ref{prop:internalm} (iii) (b) \\
		\bottomrule
	\end{tabular}
	}
	\caption{In the second column, the list of the Sugeno-like $FG$-functionals constructed in Example~\ref{ex:conivsug} satisfying the property in the first column is given. In the third column, the justification is indicated.}
	\label{tab:propivsug}
\end{table}

\section{IV Sugeno-like $FG$-functional applied in a Brain Computer Interface} \label{sec:bci_signal}

In this section, we use the proposed IV Sugeno-like $FG$-functional in a BCI framework that uses interval-valued predictions. We detail this framework, how the interval-valued logits are constructed, and how different versions the proposed IV-Sugeno compares to other IV-aggregations, and to other BCI frameworks.

\subsection{Motor Imagery Brain Computer Interface Framework} 
\label{sec:tradBCI}
BCI systems usually consist of four different modules. The initial step is the EEG data acquisition and conditioning (signal amplification and different filters to remove noise, high impedance sensors, etc.). The second block of procedures is related to extracting features. It often includes filtering the EEG signal in one or several frequency bands, which can be subject-specific or fixed for all participants \cite{akin2002comparison,Nierhaus19}. This block might include some dimensionality reduction procedure such as Spatio-spectral decomposition (SSD) \cite{nikulin2011,Haufe2014}. Another common step is the application of an optimized spatial filtering method such as common spatial patterns (CSP, \cite{guger2000real,gramfort2013meg,blankertz2008optimizing,sannelli_csp_2011,5626227}). In case of focusing on power-based features, they are usually gaussianized by applying the natural logarithm. Then, the classification step of the features is carried out, usually based on linear classifiers (\cite{MueAndBir03,Vidaurrelda,izenman2013linear}). Often classifier ensembles are implemented, where the final decision is performed combining the outputs of all classifiers \cite{fumanal2020interval,fumanal2021motor}. The last module is related to the feedback to the user, which usually is in form of visual feedback, but can also be auditory or somatosensory \cite{VidPasRamLorBlaBirMue13,vidaurre2019a,VidaurreImproving21}.

\subsection{Feature extraction and classification}
\label{sec:feature}

Features from the EEG were extracted from four subject-independent and overlapped frequency bands which cover the range from low alpha to high beta in the following ranges: 6-10, 8-15, 14-28 and 24-35 Hz. The time interval where the features were extracted was optimized for each band and class-pair by analysing the event-related desynchronization/synchronization (ERD/ERS) time courses over the sensorimotor channels \cite{Sannelli2019,Vidaurre_2020}. For that, the envelope of the EEG signals was estimated with the magnitude of the analytic band pass filtered EEG and averaged over trials of each class separately. The final time interval was chosen adding those time samples whose the pair-wise discriminability using the envelopes as features were higher. The selected time intervals were used to crop the EEG signals and then SSD was computed to reduce the number of dimensions in a specific band \cite{nikulin2011,Haufe2014}. SSD decomposes multivariate data into sources of maximal signal-to-noise ratio (SNR) in a narrow-band. The selected sources with high SNR were then projected on to a few common spatial directions \cite{,Sannelli2019,KawSamMueVid14}. The power of the projected training signals was then computed for each trial and the natural logarithm applied. Finally, LDA classifiers were trained for subsequent classification, see Figure \ref{fig:bci_scheme}.

Regarding the testing set, the features were extracted as follows: the EEG was filtered in the four narrow bands of interest and then projected onto the corresponding SSD and CSP filters (depending on the class pair and band). The data were cropped in the time intervals selected with the training set. Finally, the power of each trial and projection was estimated and the natural logarithm applied. These features were then classified using the LDA previously trained.

\begin{figure}
	\includegraphics[width=\linewidth]{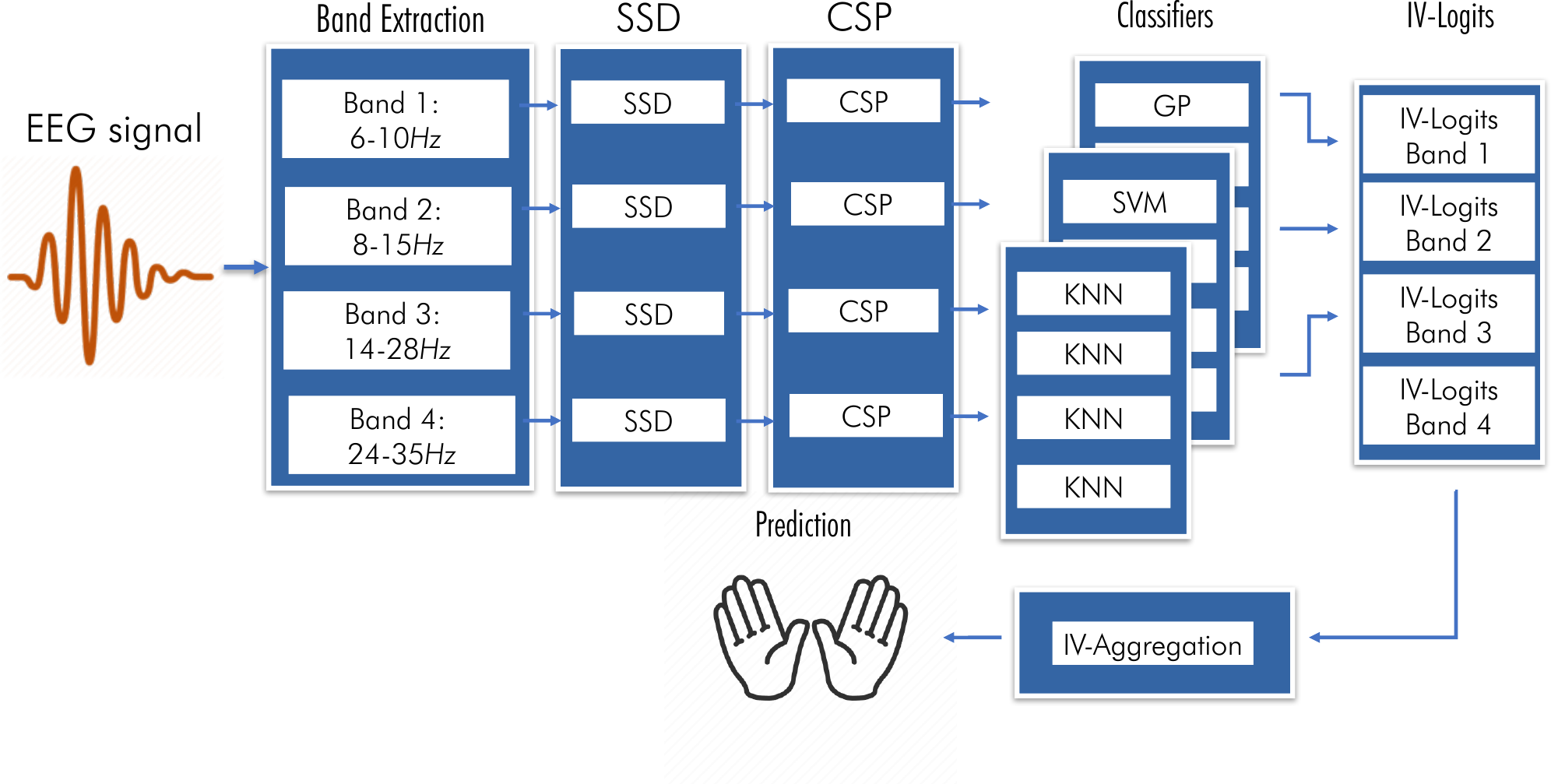}
	\caption{Graphical depiction of the BCI framework. After measuring the EEG, SSD is computed in four different bands to reduce dimensions and increase the signal-to-noise ratio. Then CSP is applyied}
	\label{fig:bci_scheme}
\end{figure}

\subsection{Interval-valued predictions using different classifiers}
In order to generate interval-valued outputs from our system, we have opted to use the predictions from different classifiers. We use the variability between them in order to estimate the uncertainty related to each of our predictions. The process is the following:

\begin{enumerate}
	\item We choose a set of $n$ different kinds of classifiers. For our experimentation, we have chosen three different types: Support Vector Machine (SVM), Gaussian Process (GP) and K-Nearest Neighbours (KNN).
	\item From the features obtained from each of the wave bands studied, we train a classifier of each kind. In our case, this means that we train a SVM, a GP and a KNN for each wave band.
	\item For each wave band, we have $n$ different predictions for each different class. We generate an interval for each wave band and class using the lowest and the maximum predictions respectively.
\end{enumerate}

After the final step, we have obtained for each class one interval-valued prediction per wave band. Then, we aggregate the interval-valued predictions for each sample using a suitable aggregation function and a $\preceq_{\alpha, \beta}$ order. Finally, we choose the maximum value according to the previously established order.

\subsection{Experimental results for a left/right hand BCI task}
We have performed our experimentation in the Clinical BCI Challenge WCCI 2020 dataset (CBCIC) \cite{chowdhury2021clinical}. This dataset consists of brain imaging signals from 10 hemiparetic stroke patients with hand functional disability in a rehabilitation task. The data contains 80 different trials of left/right hand movements. Decoding motor cortical signals of brain-injured presents several challenges as the  presence of irregular because of the altered neurodynamics. 

We have tried different versions of the newly proposed IV Sugeno-like $FG$-functionals, using the cardinality fuzzy measure in all of them, i.e., $m(A)=\frac{|A|}{n}$ for $A\subseteq N$:
\begin{itemize}
	\item IV-Sugeno$_1$: $G(X_1,\dots,X_n) = \frac{1}{n}\sum_{i=1}^{n}X_i$, $F(X, Y) = X^2  Y + X(1-Y) $;
	\item IV-Sugeno$_2$: $G(X_1,\dots,X_n) = \frac{1}{n}\sum_{i=1}^{n}X_i$, $F(X, Y) = X(1-Y) $;
	\item IV-Sugeno$_3$: $G(X_1,\dots,X_n) = \max(X_1,\dots,X_n)$, $F(X, Y) = \min(X, Y)$,  
\end{itemize}
for $X_1,\dots,X_n, X,Y\in\li$.

We have also studied other interval-valued aggregation functions: the interval-valued OWA operators and the interval-valued Moderate deviation functions. We have also compared with other BCI frameworks based on Deep Learning: the EEG net \cite{Lawhern2018}, and two other different Deep Learning architectures \cite{hbm23730}; a different version of CSP using different temporal scales to extract features \cite{8553378}; and Gradient Boosting \cite{9342132}. 

In order to evaluate the performance of each different proposal in this dataset, each participant's dataset was randomly sampled in ten different partitions (each with 50$\%$ train and 50$\%$ test trials), resulting in a total of $80$. The final performance of each configuration was obtained averaging each single dataset accuracy and F1-Score. 

The results were obtained using different aggregation functions in the decision making phase and compared the newly proposed MCAs. Both the adaptive and the non-adaptive mixing parameter were employed with a set of standard aggregations and also with the already existing penalty-based aggregation functions.

Table \ref{tab:iv_aggs} shows the results for each of the different interval-valued aggregations used in this BCI framework. IV-Sugeno$_1$ obtained the best result, followed by IV-Sugeno$_2$ and IV-OWA$_2$. Table \ref{tab:full_cmp} shows the comparison between the best interval-valued aggregation using our BCI framework and other BCI frameworks. The IV-Sugeno obtained the best results in this comparison, followed by the Multiscale CSP.

\begin{table}
	\centering
	\begin{tabular}{ccc}
		\toprule
		Aggregation function & Accuracy & F1-Score \\
		\midrule
		IV-Sugeno$_1$ & $\mathbf{0.8175  \pm 0.1342}$ &  $\mathbf{0.8149 \pm 0.1366}$ \\
		IV-Sugeno$_2$ & $0.8162  \pm 0.1338$ &  $0.8138 \pm 0.1359$ \\
		IV-Sugeno$_3$ & $0.7994  \pm 0.1324$ &  $0.8014 \pm  0.1313$ \\
		\midrule
		IV-OWA$_1$ & $0.8141  \pm 0.1327$ & $0.8127 \pm 0.1340$\\
		IV-OWA$_2$ & $0.8162  \pm 0.1338$ & $0.8138 \pm 0.1359$\\
		IV-OWA$_3$ & $0.8100  \pm 0.1327$ & $0.8073 \pm 0.1352$\\
		\midrule
		IV-MD & $0.8097 \pm 0.1318$ & $0.8085 \pm 0.1315$ \\
		\bottomrule
		
	\end{tabular}
\caption{Results for the CBCIC dataset using different interval-valued aggregations.}
\label{tab:iv_aggs}

\end{table}

\begin{table}
	\centering
	\begin{tabular}{ccc}
		\toprule
		Framework & Accuracy & F1-Score \\
		\midrule
		IV-Sugeno$_1$ & $\mathbf{0.8175  \pm 0.1342}$ &  $\mathbf{0.8149 \pm 0.1366}$ \\
		EEG Net \cite{Lawhern2018}& $0.6562 \pm 0.1232 $ & $0.5933 \pm 0.1712$\\
		Shallow Net \cite{hbm23730} & $0.7453 \pm 0.13289$ & $0.7342 \pm 0.1489$\\
		Deep Net \cite{hbm23730} & $0.5331 \pm 0.1356$& $0.4218 \pm 0.1282$\\
		Multiscale CSP \cite{8553378} & $0.7956 \pm 0.1144 $& $0.7911 \pm 0.1175$\\
		Gradient Boosting \cite{9342132} & $0.5956 \pm 0.1203 $ & $0.5354 \pm 0.1169$\\
		\bottomrule
		
	\end{tabular}
	\caption{Results for the CBCIC dataset using different BCI frameworks.}
	\label{tab:full_cmp}	
\end{table}

\section{IV Sugeno-like $FG$-functional applied in Social Network Analysis} \label{sec:social_network}

In this section we propose an interval-valued version of the affinity functions proposed in \cite{idocin2020borgia} to characterize the relationship between two actors in a network. 

We start by recalling the notions of centrality measure and affinity function in social network analysis. Then we show how we construct the iv-affinity functions and propose new centrality measures based on them. These metrics characterize each actor based on the difference in commitment that has in its relationships.

\subsection{Centrality measures in social network analysis}
In social network analysis, centrality measures are metrics to ponder how relevant each node is in a structure \note{\cite{landherr2010critical, newman2018networks}}. Some very well known centrality measures are:

\begin{itemize}
	\item Degree centrality: the number of edges incident upon an actor. In the case of directed networks, the degree is the sum of the number of edges incident to the actor (in-degree) and the number of edges salient to the actor (out-degree).
	\item Betweenness centrality: the betweenness of an actor is the number of times that node is in the shortest path of other two nodes. It measures the importance of the actor in the network's information flow.
	\item Closeness centrality: the closeness centrality of an actor is the average length of the shortest path between that node and the rest of the nodes in the network. It establishes a center-periphery difference.
	\item Eigenvector centrality: it assigns a relative score to each actor in the network based on the idea that connections to well connected actors should ponder more than connections to poorly connected ones.
\end{itemize}

\subsection{Affinity functions in social network analysis}
Affinity functions were proposed in \cite{idocin2020borgia} as a way to measure the relationship between a pair of actors in a social network using their local information. We denote by $V$ the set of all actors. ``Affinities" are defined as functions over the set $V^2$ of all pairs of actors in a given social network assigning a number $F_C(x,y) \in [0, 1]$  to every pair of actors $(x,y)\in V^2$. 

Usually, $C$ is the adjacency matrix of the network. Each of the entries $C(x, y)$ in $C$ quantifies the strength of the relationship for the pair of actors $x, y$ in a network $V$. The affinity between two actors shows how strongly they are connected according to different criteria, depending on which aspect of the relationship we are taking into account.
A $0$ affinity value means that no affinity has been found at all while an $1$ value means that there is a perfect match according to the analyzed factors.

In the following, we recall  the definition of two affinity functions (additional affinity functions can be found in \cite{idocin2020borgia} and \cite{fumanal2021concept}):

\begin{itemize}
	\item Best friend affinity: it measures 
	the importance of a relationship with an agent $y$ for the agent $x$, in relation to all the other relationships of $x$: 
	\begin{equation} \label{eq-bf}
	F^{BF}_C (x,y)= \frac{C(x,y)}{\sum_{a \in V} C(x,a)}. 
	\end{equation}  
	
	\item Best common friend affinity: it measures the importance of the relationship taking into account how important are the common connections between the connected nodes to $x$ and $y$, in relation to all other relationships of $x$ in the network:
	\begin{equation} \label{eq-bcf}
	F^{BCF}_C (x,y)= \frac{\max_{a \in V }\{\min \{ C(x,a), C(y,a)\} \} }{\sum_{a \in V} C(x,a)}. 
	\end{equation}
	
\end{itemize}

For example, in the case of Best common friend affinity, $0$ value means that there is no common connections between the two actors and a high value means that their shared friends are important to both of them. Since affinities are not necessarily symmetrical, the strength of this interaction depends on who the sender and receiver are, as it happens in human interactions e.g. unrequited love. So, it is possible that actor $x$ has an affinity of value $1$ with $y$, while $y$ has a much lower value with $x$.

\subsection{Interval-valued affinity functions}

We define interval-valued affinity functions (IV-affinity functions) as functions that characterize the relationship between two actors, $x, y$ with an interval in the $[0,1]$ range, where the width of that interval represents the difference in commitment between two parties. We construct an interval-valued affinity function using a previously computed numerical affinity function. Then, the interval is constructed as:
\begin{equation}
F_{C,IV}(x, y) = \left[\min \{F_C(x, y), F_C(y, x)\},  \max \{F_C(x, y), F_C(y, x)\}\right]
\end{equation}
Because $F_C(x, y) \neq F_C(y, x)$ in most relationships, this means that in most of them the actors give different levels of commitment than their counterparts. In real life, these kind of situations are usually solved by finding a compromise between both parties. The IV-affinity function models this idea, representing the actual relationship that it is formed with an interval that ranges from both levels of commitment. The interval models the fact that we know that the final compromise achieved by both actors should be between both commitment levels, but we do not know exactly the final compromise.

One of the main difference between IV-affinity functions and their numerical counterpart is that they are symmetric. This can be convenient, as it allows to represent the relationship between two parties with one interval instead of two numerical values. \note{This also opens the possibility of using some of the existing methods that require symmetric matrices in social network analysis \cite{newman2018networks, shi2015community}, while} retaining the desirable properties of affinity functions, i.e. zeros-sum game philosophy or local-only interactions taken into account \note{\cite{reddy2009role}}.

\subsection{Using interval-valued affinity functions to construct centrality measures}
In this section we present four different centrality measures using IV-affinity functions. \note{Usually, centrality measures ponder the importance of each actor in the network. However, our proposed} centrality measures characterize the tendency of each actor to form relationships that have very different levels of commitment, and if the actor tends to show more or less commitment than the other party in each of its relationships.

The proposed centrality measures are:
\begin{enumerate}
	\item \textbf{Asymmetry} is the tendency of the actor to form relationships with different levels of commitment.
	\item \textbf{Altruism} is the tendency of the actor to form relationships in which its level of commitment is bigger than the other party. 
	\item \textbf{Egoism} is the tendency of the actor to form relationships in which its level of commitment is lesser than the other party. 
	\item \textbf{Generosity} is the difference between altruism and egoism. A positive generosity means that overall, the actor tends to give more commitment in its relationships than the other part. A negative generosity means that the actor tends to give less commitment than the other part in a relationship.
\end{enumerate}

We have considered two possibilities to compute these metrics: using the width of each IV-affinity function and aggregate those values, or aggregating the IV-affinity values using an IV-aggregation and then use the width of the aggregated interval. We have opted for the latter, because in this way we are taking into account that some intervals have the same width, but are different. 

We are specially interested in this property in the case where two actors have the same average width for their respective IV-affinity values. In these occasions, we consider that the one that presents more variety in their IV-affinities is more asymmetric than the other. For example, if an actor $x$ have $3$ different IV-affinities, and all of the are $(0,0.3)$, and we have other actor $y$ that has as IV-affinities $(0,0.3)$, $(0.4, 0.7)$ and $(0.7, 1.0)$. We should assign the highest asymmetry value to $y$, because this actor presents this level of asymmetry in its relationships in all the $[0,1]$ range.

Taking these considerations into account, we compute these centrality measures using the following expressions, with an IV-aggregation function $I$. 

\begin{enumerate}
	\item \textbf{Asymmetry}: 
		\begin{equation}
				A(x) = w\left(I(F_{C,IV}(x, y_1),\dots,F_{C,IV}(x, y_m) )\right), 
		\end{equation}
	where $\{y_1,\dots,y_m\} = V\setminus\{x\}$.
	\item \textbf{Altruism}: 
		\begin{equation}
				L(x) = w\left(I(F_{C,IV}(x, y_1),\dots,F_{C,IV}(x, y_k) )\right),  
		\end{equation}
	where   $\{y_1,\dots,y_k\} \subseteq V\setminus\{x\}$ is the set of all actors $y_j$ fulfilling $F_C(y_j, x) \leq F_C(x, y_j)$ for each $j=1,\dots,k$. 
	\item \textbf{Egoism}: 
		\begin{equation}
						E(x) = w\left(I(F_{C,IV}(x, y_1),\dots,F_{C,IV}(x, y_l) )\right),  
		\end{equation}
	where   $\{y_1,\dots,y_l\} \subseteq V\setminus\{x\}$ is the set of all actors $y_j$ fulfilling $F_C(x, y_j) \leq F_C(y_j,x)$ for each $j=1,\dots,l$. 
		
	\item \textbf{Generosity}: 
		\begin{equation}
		S(x) = L(x) - E(x)
		\end{equation}	
\end{enumerate}

For our experiment, we have used as $I$ the proposed IV Sugeno-like $FG$-functional, 
with $G(X_1,\dots, X_n)=\min\{{\bf 1}, \sum_{i=1}^{n} X_i\}$, $F(X, y) = X \cdot (1- y)$ and the cardinality measure $m(A)=\frac{|A|}n$ as the fuzzy measure. Note that in the three respective cases of centrality measures defined above, the arity $n$ of applied IV Sugeno-like $FG$-functional $I$ is $m,k,l$, respectively.

\subsection{Experimental results for a social network}

		%
In this section we have studied the proposed centrality measures in a word association network, constructed using \emph{The Younger Edda} book. \emph{The Younger Edda} is Old Norse textbook of mythical texts, written approximately in 1220 by Snorri Sturluson. This book contains the tales of popular characters in the Nordic folklore, like Odin, Thor or Loki.

In order to extract the word association network from this text, we have followed the standard procedure to tokenize and lemmatize the text \cite{webster1992tokenization, fox1992lexical} and we have used a pre-trained multilayer perceptron in the Python Natural Language ToolKit \cite{bird2009natural} to purge every word that is not a noun, since we only want to model interaction between entities and concepts. \note{Once we have extracted the nouns from the text, what we have is a series of stemmed tokens. To obtain a network, we need to determine its nodes and edges. In the case of the nodes, we make a bijective association, so that one noun corresponds to one node, and vice versa. There are different ways to compute the edges in terms of noun co-occurrence. We have decided to create an edge every time a word appears in a k-distance or less from another in the text, choosing k as 10.}
The text has been obtained from the Gutenberg Project \cite{stroube2003literary}.

We have computed the proposed centrality measures in this network: asymmetry, egoism, altruism and generosity. Figure \ref{fig:edda_measures} shows the results in this network, coloring each node according the value of each metric. 
We found that the actors with highest degree tend to be less altruistic, like ``Odin'', ``Thor'', ``Balder'' or ``Loki''. However, low values in altruism do not necessarily imply high egoism values, as these same actors only showed moderate egoism values. The highest egoism values were located in actors like ``gold'' or ``gods'', that connected low degree actors with the ones with highest value. These actors also presented a moderate altruism value, that comes from their relationships with higher degree values. This means that the most egoist actors are willing to have a small number of connections where they are the ``losing'' part, because it allows them to have more relationships where they are the ``favored'' part.

Regarding generosity, most nodes have a negative value. The generous actors are low degree actors connected to higher degree actors, which in this case are concepts not very frequent in the original text. The less generous actors have small-medium degree values, generally connected to a small set of higher degree, high egoism actors, and a another set of very generous ones.

Table \ref{tab:measures} shows the top values for each centrality measure. The highest asymmetry values correspond to actors like ``Hammer'', mostly connected with ``Thor'', a high degree actor; ``Journey'', connected to individual characters; and other general concepts. Altruism presents a similar top 10 of values, but with the presence of two individuals: ``Sigurd'', who is connected to one of the most egoist actors, ``gold''; \note{``Hermod'', connected almost exclusively to characters with negative generosity like ``Balder'', ``Odin'' and ``\AE sir''; and ``Hymir'', who is connected to a similar set of actors to ``Hermod''}. The most egoist actors are general concepts like `Name'', ``Man'', ``Land'' or ``Gold''. The only exception is ``Atle'' because he has only one relationship where he is the least committed part.

\begin{figure}
	\begin{tabular}{cc}
		\includegraphics[width=.5\linewidth]{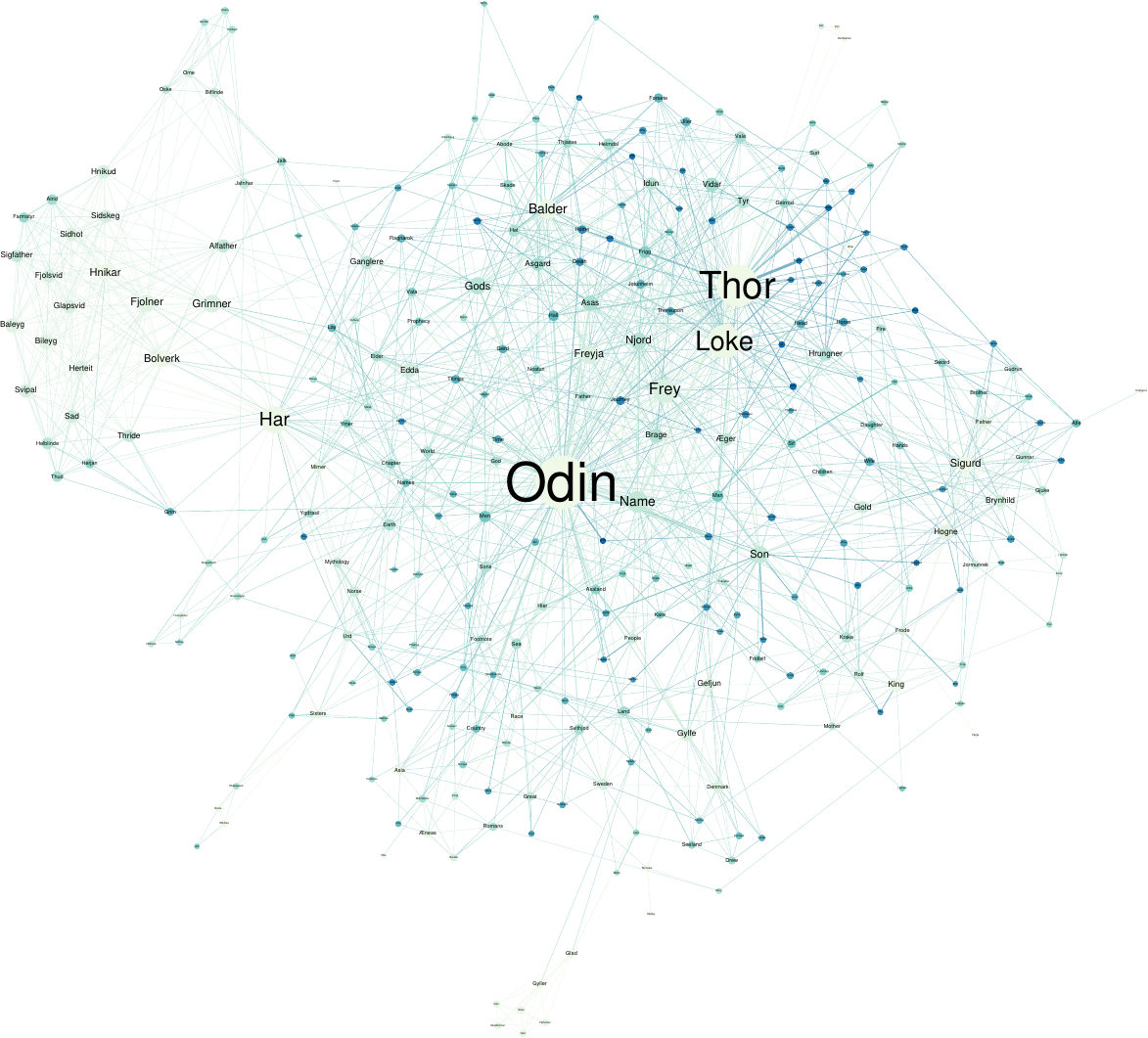} & \includegraphics[width=.5\linewidth]{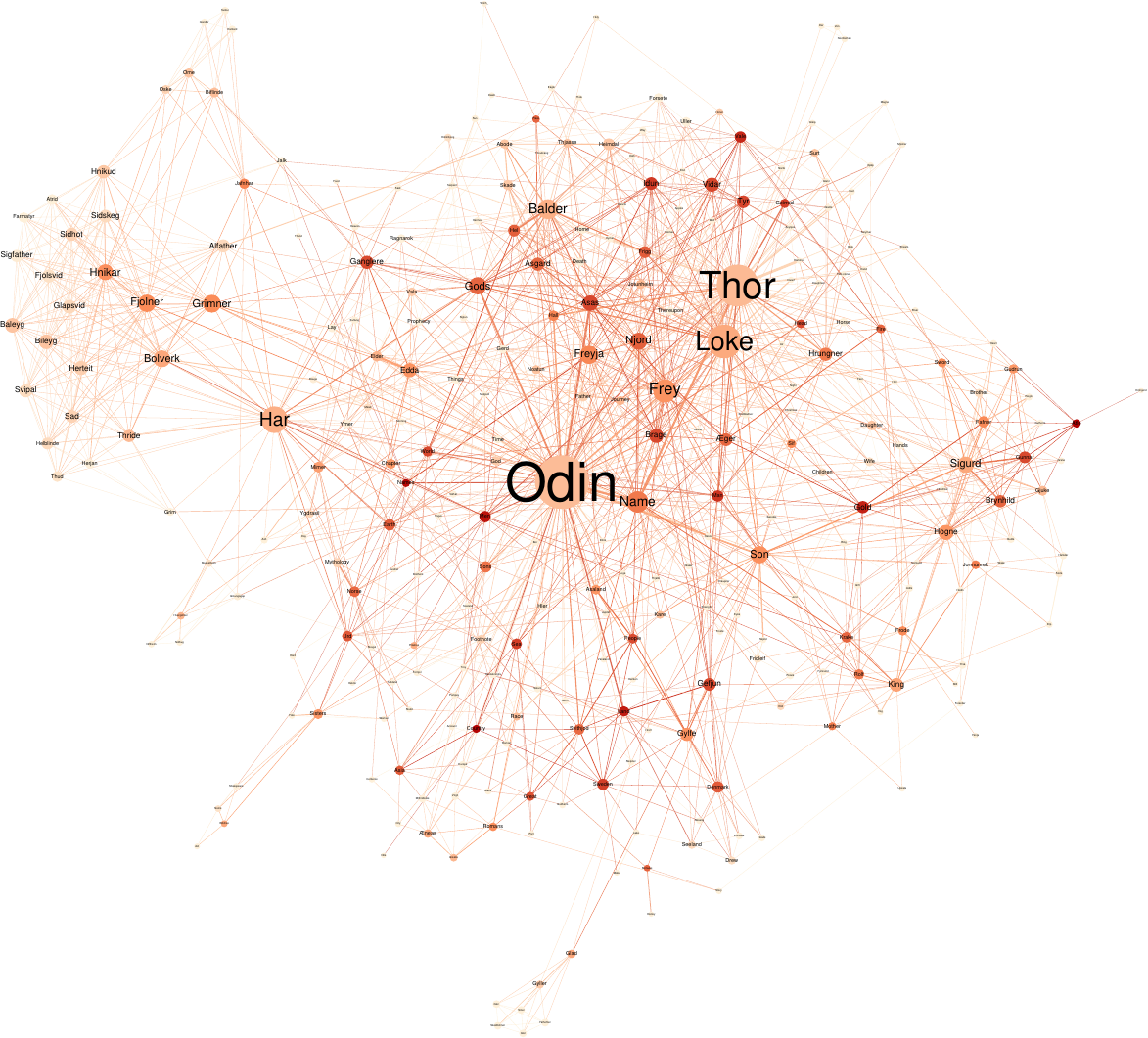} \\
		(a) Altruism & (b) Egoism\\
		\includegraphics[width=.5\linewidth]{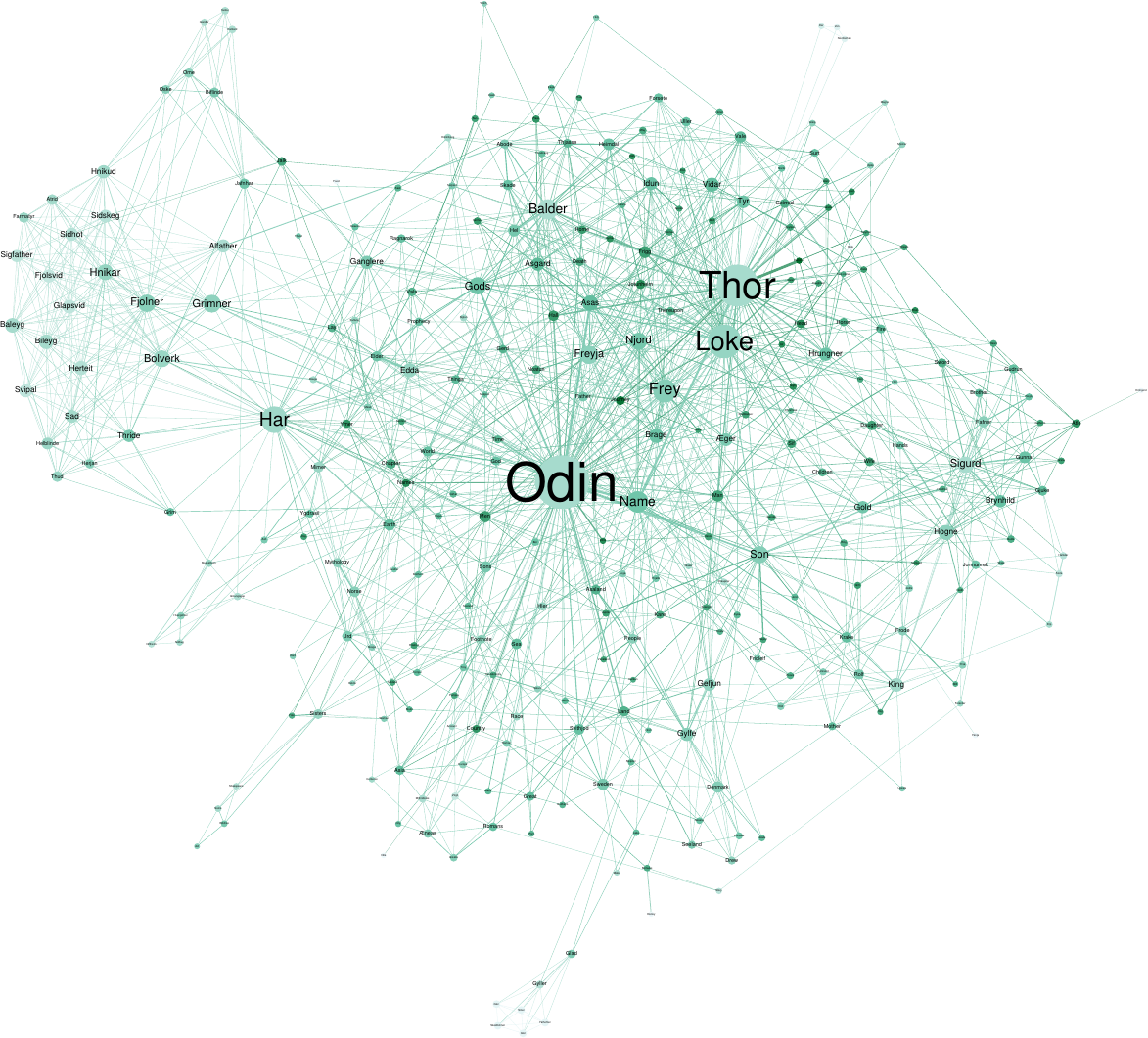} & \includegraphics[width=.5\linewidth]{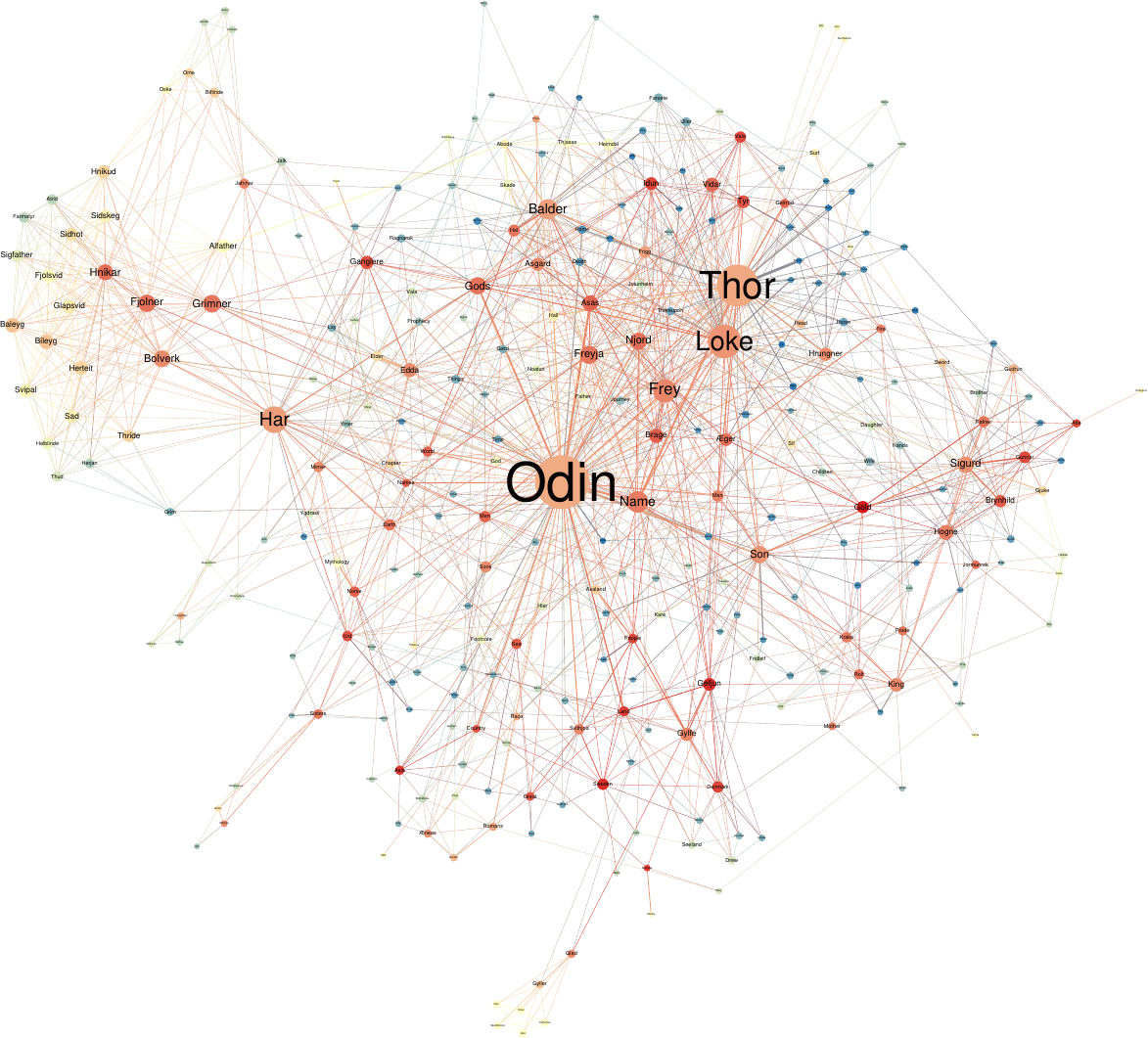} \\
		(c) Asymmetry & (d) Generosity \\
	\end{tabular}
\caption{Altruism (a), Egoism (b), Asymmetry (c) and Generosity (d) for the \emph{Younger Edda} network, marked with different colors. Node size is proportional to the node degree. We can see that most actors are more egoist than altruist. In fact, some of the most important actors in this network show no altruism at all. However, no altruism does not necessarily imply a high egoism. For example, ``Odin'' and ``Thor'' show no altruism (so all their in-affinities are higher or equal than their out-affinities), but their asymmetry value is also low.}
\label{fig:edda_measures}
\end{figure}

\begin{table}
	\centering
	\begin{tabular}{cc}
		\toprule
		\multicolumn{2}{c}{Assymetry} \\
		Actor & Value \\
		\midrule
		Hammer & 0.7808 \\
		Journey& 0.7512 \\
		Drink& 0.6451 \\
		Names& 0.6402 \\
		Air& 0.6391 \\
		Night& 0.6312 \\
		River& 0.62891 \\
		Jotunheim& 0.6257 \\
		Head	& 0.6217 \\
		Giants & 0.6188 \\
		\bottomrule
	\end{tabular}
	\begin{tabular}{cc}
		\toprule
		\multicolumn{2}{c}{Altruism} \\
		Actor & Value \\
		\midrule
		Drink & 0.6451 \\
		Air& 0.6391 \\
		Night & 0.6312  \\
		River& 0.6289\\
		Hammer& 0.6255\\
		
		Giants& 0.6154\\
		Oath& 0.6154\\
		\note{Hymir}& 0.5883\\
		\note{Sigurd}& 0.5865\\		
		Hermod& 0.5818\\
		\bottomrule
	\end{tabular}\\
	\begin{tabular}{cc}
		\toprule
		\multicolumn{2}{c}{Egoism} \\
		Actor & Value \\
		\midrule
		Atle & 0.7875 \\
		Country & 0.7777\\
		Names & 0.7678\\
		Men &  0.73320\\
		Land & 0.7103\\
		
		Gold & 0.6920 \\
		Vale & 0.6748 \\
		Man & 0.6443 \\
		Idun & 0.6261 \\
		Geirrod &	 0.6124 \\	
		\bottomrule
	\end{tabular}
	\begin{tabular}{cc}
		\toprule
		\multicolumn{2}{c}{Generosity*} \\
		Actor & Value \\		
		\midrule
		Home & 0.5226 \\
		Wife & 0.3484 \\
		Journey & 0.3277 \\
		Children & 0.2550 \\
		Brother & 0.2210 \\
		\midrule
		People & -0.4730 \\
		Idun & -0.4830 \\
		Land &  -0.4985\\
		Gefjun & -0.5414 \\
		Gold & -0.6000 \\
		\bottomrule
	\end{tabular}
\caption{Top 10 values for the different centrality measures proposed: Asymmetry, Altruism, Egoism, Generosity. *We showed the generosity on the top 5 actors that showed $>0$ altruism, because when altruism $= 0 \to $ generosity $=$ egoism, and the analogous thing for the least 5 ones with respect to egoism.}
\label{tab:measures}
\end{table}

\section{Conclusions and Future Lines} \label{sec:conclusions}
In this work we have proposed a new version of the generalized Sugeno integral to aggregate interval-valued data (IV-Sugeno). This function is designed to aggregate intervals taking into account the coalitions between the input data, just as the numerical Sugeno integral. We have also proposed two applications in which we use interval-valued data: a brain computer interface framework where the intervals measure the uncertainty related to the output of different classifiers; and social network analysis, where the intervals measure the difference in commitment in a relationship, \note{and how this can be used to construct a centrality measure to characterize each actor.}

\note{Our results show that the generalized IV-Sugeno aggregation performs better than other IV-aggregations for a brain computer interface classification network. We have also found that the functions used to construct the generalized IV-Sugeno are critical in its performance and its mathematical properties. Finally, we have also shown how the generalized IV-Sugeno can be used to successfully characterize each actor in a network depending on the asymmetry in its relationships.}

\note{Future research shall study the use of the proposed IV-Sugeno in other settings, like image processing; and the study of other fuzzy integrals in the interval-valued setting. We also intend to study the possibilities of using IV-affinity functions with classical centrality measures in social network analysis and the possible correlations of the newly centrality measures proposed with classical ones.}

\section{Acknowledgements}

This work was supported by the Spanish Ministry of Economy and Competitiveness through the Spanish National Research (project PID2019-108392GB-I00/ financed by MCIN/AEI/10.13039/501100011033 and the grant  VEGA 1/0267/21.

\bibliographystyle{elsarticle-num}
\bibliography{iv_sugeno}

\end{document}